\documentclass[a4paper, 10pt]{article}
\usepackage[T1]{fontenc}
\usepackage{amsmath}
\usepackage{amsfonts}
\usepackage{amssymb}
\usepackage[latin9]{luainputenc}
\usepackage{mathrsfs}
\usepackage{bbm}
\usepackage{subfig}
\usepackage[left=3.5cm, right=3.5cm]{geometry}
\usepackage{amsthm}
\usepackage{graphicx}

\usepackage{url}

\usepackage{color}
\definecolor{lightgray}{gray}{0.75}

\newtheorem{theorem}{Theorem}[section]
\newtheorem{corollary}[theorem]{Corollary}

  \DeclareMathOperator{\stderr}{stderr}
  \DeclareMathOperator{\tr}{tr}
  \DeclareMathOperator{\var}{var}

  \DeclareMathOperator{\SR}{SR}

 \DeclareMathOperator{\implied}{implied}
  \DeclareMathOperator{\Lo}{Lo}
  \DeclareMathOperator{\Mertens}{Mertens}

\newcommand{\be}{\begin{equation}}
\newcommand{\ee}{\end{equation}}



\begin{document}

\title{Optimal Dynamic Strategies on Gaussian Returns}

\author{Nick Firoozye and Adriano Koshiyama}


\maketitle

\begin{center}
Department of Computer Science\\
University College London \\
\texttt{n.firoozye@ucl.ac.uk, a.koshiyama@cs.ucl.ac.uk}\\
\date{\today}

\end{center}

\begin{abstract}
Dynamic trading strategies, in the spirit of trend-following or mean-reversion, represent an only partly understood but lucrative and pervasive area of modern finance. Assuming Gaussian returns and Gaussian dynamic weights or {\sl signals}, (e.g., linear filters of past returns, such as simple moving averages, exponential weighted moving averages, forecasts from ARIMA models), we are able to derive closed-form expressions for the first four moments of the strategy's returns, in terms of correlations between the random signals and unknown future returns. By allowing for randomness in the asset-allocation and modelling the interaction of strategy weights with returns, we demonstrate that positive skewness and excess kurtosis are essential components of all positive Sharpe dynamic strategies, which is generally observed  empirically;  demonstrate that {\sl total least squares} (TLS) or {\sl orthogonal least squares} is more appropriate than OLS for maximizing the Sharpe ratio, while {\sl canonical correlation analysis} (CCA) is similarly appropriate for the multi-asset case;  derive standard errors on Sharpe ratios which are tighter than the commonly used standard errors from Lo; and derive standard errors on the skewness and kurtosis of strategies,  apparently new results. We demonstrate these results are applicable asymptotically for a wide range of stationary time-series.


{{\em MSC Numbers:} 60G10, 62E15, 62P05, 62F99, 91G70, 91G80}

\end{abstract}





\section{Introduction}

CTAs ({\em Commodity Trading Advisors}) or managed-future accounts are a subset of asset managers with  over \$341bn of assets under management \cite{BarclayHedge} as of Q2 2017. The predominant strategy which CTAs employ is trend-following. Meanwhile, bank structuring desks have devised a variety of {\em risk-premia} or {\em styles} strategies (including momentum, mean-reversion, carry, value, etc) which have been estimated to correspond to between approximately \$150bn \cite{Miller} to \$200bn \cite{Allenbridge} assets under management. Responsible for over 80\% of trade volume in equities and a large (but undocumented amount due to the OTC nature) of the FX market, \cite{HFT}, high-frequency trading firms (HFTs) and e-trading desks in investment banks are known to make use of many strategies which are effectively short-term mean-reversion strategies. In spite of the relatively large industry undergoing  recent significant growth, a careful analysis of the statistical properties of strategies, including their optimisation, has only been undertaken in relatively limited contexts.

\begin{figure}[h!]
	\centering
	\includegraphics[width=\linewidth]{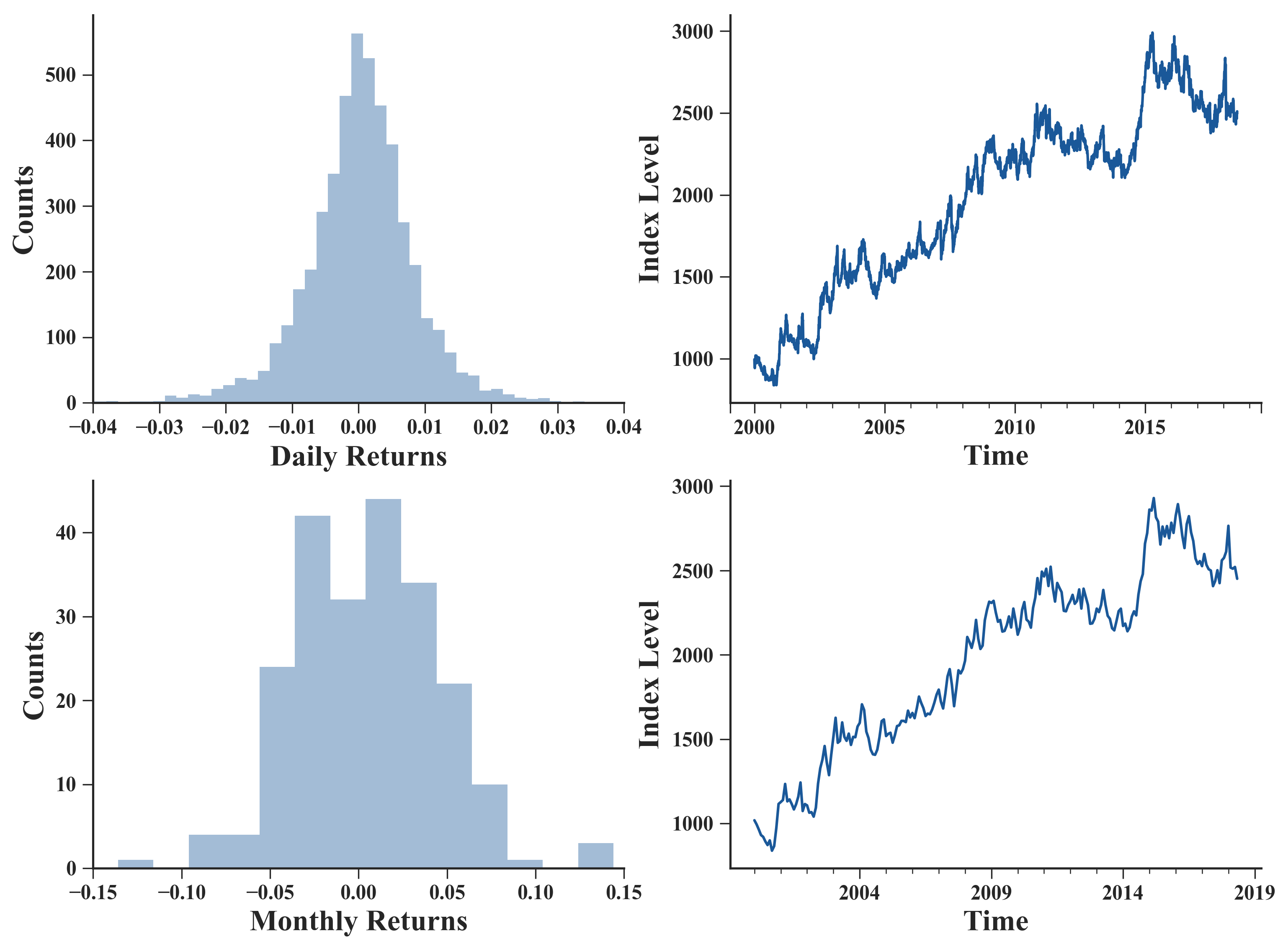}
	\caption{{\bf SocGen Trend Followers Index:} daily returns and monthly returns profiles} \label{}
\end{figure}

The corresponding statistics for the SG Trend index area in the table below and except for some noise show that skewness and excess kurtosis are laregly positive for CTAs.

\begin{table}[h!]
  \centering
    \caption{\bf Soc Gen Trend Index, Daily and Monthly Statistics}
    \begin{tabular}{l|cc}
      $\;$ & \textbf{Daily} & \textbf{Monthly}\\ 
      \hline
      Ann Avg Return (\%) & 5.695 & 5.752 \\ 
      Volatily (\%)&  13.283&14.088\\
      Sharpe Ratio&  0.429 & 0.408\\
      Skewness & -0.448& 0.186\\
      Exc Kurtosis & 3.845 & 0.807\\
      
   \end{tabular}
    \label{tab:table1}
\end{table}


Algorithmic trading strategies we consider are time-series strategies, often divided into mean-reverting or reversal strategies, trend-following or momentum strategies, and value strategies (also sometimes known as mean-reversion).\footnote[1]{Other common strategies include carry and short-gamma or short-vol. Unlike mean-reversion, momentum, and value, these do not rely on the specifics of the auto-correlation function}. Each such time-series related strategy is a form of signal processing.
In more standard signal processing, the major interest is in the de-noised or smoothed signals and their properties. In algorithmic trading, the interest is instead in the relationship between statistics like the moving average or some other form of smoothed historic returns (unfortunately, usually termed the {\sl signal}) and the unknown future returns. We show that when we consider both to be random variables, it is actually the interaction between these so-called signals and future returns which determines the strategy's behaviour.

Equities, and in particular SPX is known to mean-revert over short horizons (e.g., shorter than 1m, typically on the order of 5-10 days) and trend only over longer horizons (i.e., 3m-18m), and mean-revert again over even longer horizons (i.e., 2y-5y) as has been well-established by the quant equities literature following on the study of \cite{Jegadeesh}  and the work of \cite{FamaFrench}. This distinct form of behaviour, with reversals on small-scale, trend on an intermediate and reversion on a long scale, is frequently observed across a large number of asset classes and strategies can be designed to take advantage of the behaviour of asset-prices across each time-scale.

Our initial goal is to find a signal, $X_t$ usually a linear function of historic log-(excess) returns $\{R_t\}$ which can be used as a dynamic weight for allocating to the underlying asset on a regular basis. We assume log-price $P_t=\sum_1^t R_{k}$. Examples of commonly used signals for macro-traders (CTAs, and other trend followers) include:

\begin{itemize}
	\item Simple Moving Average (SMA): $$X_t=\frac{1}{T}\sum_1^T R_{t-k}$$
	\item Exponentially-Weighted Moving Average (EWMA): $$X_t=c(\lambda)\sum_{k=1}^\infty \lambda^k R_{t-k}$$
	\item Holt-Winters (HW, or double exponential smoothing) with or without seasonals, Damped HW
	\item Difference between current price and moving average:\footnote[2]{We note that if we replace $P$ by $\log(P)$ and $R_t=\log(P_t)-\log(P_{t-1})$, this filter amounts to $X_t=\sum \frac{T-k}{T} R_{t-k}$, i.e., a {\em triangular} filter on returns, which bears some similarity to EWMA on returns.}
	$$X_t= P_{t-1} - {1\over T}\sum_1^T P_{t-k}$$
	\item Forecasts from ARMA(p, q) models: $$X_t = \phi_1 R_{t-1} + ... + \phi_p R_{t-p} + \theta_1 \varepsilon_{t-1} + ... + \theta_q \varepsilon_{t-q} $$
	\item Differences between SMAs: $$X_t= {1\over {T_1}}\sum_1^{T_1} P_{t-k} - {1\over {T_2}}\sum_1^{T_2} P_{t-j}$$
	\item Differences between EWMAs: $$X_t=c(\lambda_1)\sum \lambda_1^k R_{t-k}-c(\lambda_2)\sum \lambda_2^k R_{t-k}$$

\end{itemize}
and variations using volatility or variance weighting such as z-scores (SMAs or EWMAs weighted by a simple or weighted standard deviation, see \cite{HarveyVol}), and transformations of each of the signals listed above (e.g. allocations depending on sigmoids of moving averages, reverse sigmoids, Winsorised signals, etc.). Other signals commonly used  in equity algorithmic trading include economic and corporate releases, and sentiment as derived from unstructured datasets such as news releases.

The returns from algorithmic trading strategies are well documented (see, e.g., \cite{ValueMomentum}, \cite{Baltas}, \cite{Hurstetal} and \cite{Lemperiere}).
Although many methods have been used to derive signals by practitioners, (see, e.g., \cite{BruderTrend} for a compendium), many of these methods are equally good (or bad) and it makes 
little practical difference whether one uses ARMA, EWMA or SMA as the starting point for a strategy design (see e.g., \cite{Levine-Pedersen}).
In this paper, we only touch on normalised signals (e.g., {\em z-scores}) and strategy returns, leaving their discussion for a subsequent study. We meanwhile note that the spirit of this paper's results carry through for the case of normalized signals and strategy returns.

Frequently, exponential smoothers have been the effective {\em best} models in various economic forecasting competitions (see, e.g., the results of the first three M-competitions \cite{Makridakis}), showing perhaps that their simplicity bestows a certain robustness, and their original intuition was sound even if the statistical foundation took a significant time to catch up.  In fact, EWMA and HW can both be justified as state-space models (see \cite{Hyndman}), and this formulation brings with it a host of benefits from mere intellectual satisfaction to statistical hypothesis tests, change-point tests, and a metric for goodness-of-fit. Exponential smoothing with multiplicative or additive seasonals and dampened weighted slopes are used to successfully forecast a significant number of  economic time-series (e.g., inventories, employment, monetary aggregates). EWMA (and the related (S)MA), and HW remain some of the most commonly used filtering methods for CTAs and HFT shops.


In the case of returns which are normal  with fixed autocorrelation function (ACF), i.e., those which are covariance stationary, signals created from  linear combinations of historic returns are indeed normal random variables which are jointly normal with returns. External datasets (e.g., unstructured data, corporate releases), are less likely  to contain normally distributed variables although there is an argument for asymptotic normality. Irrespective, our approach is to assume normality of both returns and signals as a starting point for further analysis.
%

While there is significant need for further study, there have nonetheless been a number of empirical and theoretical results of note in this area. Fung and Hsieh were the first to look at the empirical properties of momentum strategies \cite{Fung-Hsieh}, noting (without any theoretical foundation) the resemblance of strategy returns to straddle pay-offs.\footnote[3]{Or as they claimed, the returns of trend following resemble those of an extremely exotic option (which is not actually traded),  daily-traded ``look-back straddles.''}
Potters and Bouchaud \cite{Potters-Bouchaud} studied the significant positive skewness of trend-following returns, showing that for successful strategies, the median profitability of trades is negative. The empirical returns of dynamic strategies are far from normal, and common values for skewness and kurtosis for single strategies can have skewness in the range of $[1.3,1.7]$ and kurtosis in the range $[8.8,15.3]$ respectively (see \cite{Hoffman}).

Bruder and Gaussel \cite{Bruder} and \cite{Hamdan} (see Appendix 2 for a superlative use of SDE-bassed methods for analyzing a wide variety of dynamic strategies) used SDEs to study the power-option like behaviour of pay-offs. Martin and Zou considered general but IID discrete time distributions (see \cite{Martin-Zou} and \cite{Martin-Bana}) to study the {\sl term-structure} of skewness over various horizons and the effects of certain non-linear transforms on the term structure of return distributions. More reBcently, Bouchaud et al \cite{Bouchaud} considered more general discrete-time distributions to study the convexity of pay-offs, and the effective dependence of returns on long-term vs short-term variance. Other studies have focused predominantly on the empirical behaviour of returns, the relationship to macro-financial conditions, the persistence of trend-following returns, and the benefits from their inclusion into broader portfolios.

In the larger portion of the theoretical studies, the assumptions have been minimal in order to consider more general return distributions.  Due to their generality, the derived results are somewhat more restrictive. Rather than opting for the most general, we choose more specific distributional assumptions, in the hope that we can obtain broader, possibly more practical results. Aside from this current study, the authors have extended this work further to consider the endemic problem of over-fitting (see \cite{Firoozye2}), proposing total least squares with covariance penalties as a means of model-selection, showing their outperformance to standard methods, using OLS with AIC.
%

In this paper, we consider underlying assets with stationary Gaussian returns and a fixed auto-correlation function (i.e., they are a discrete Gaussian process). While we make no
defence of  the realism of using normal returns, we find that normality
can be exploited in order to ensure we understand how the returns of linear
and non-linear strategies should work in theory and to further the understanding of the interaction between properties of returns and of the  {\sl signals} as a basis for the development and analysis of dynamic strategies in practice.

Given a purely-random mean-zero covariance-stationary discrete-time Gaussian process for returns,  the signals listed above, whether a EWMA or an ARMA forecast, can be expressed as convolution filters of past returns, i.e., our signal $X_t$ can be expressed as
$$ X_t = \sum_{k\geq 1} \phi(k) R_{t-k}$$
This is an example of a time-invariant linear filter of a Gaussian process. If we restrict our attention to those filters for which are square summable,  $\sum_1^\infty\phi(k)^2<\infty$, then it is well-known that the resulting filtered series is also Gaussian and jointly Gaussian with $R_t$.

Our underlying premise is that the important distributions to consider for the analysis of dynamic strategies is a product of Gaussians (rather than a single Gaussian as would usually apply in asymptotic analysis of asset returns). This product measure can be justified on many levels and we discuss large sample approximations in the appendix.

The resulting measure which determines the success of the strategy is the correlation between the returns and the signals, a measure which, in the context of measuring an active manager's skill is known as the {\em information coefficient} or IC as given in the {\em Fundamental Law of Active Management} detailed in \cite{GrinoldKahn}. While there is a large body of literature on the IC and its relationship to information ratios, (see for example \cite{Lee} for formulas similar to equation (\ref{Eqn:Sharpe})), the derivations, resulting formulae and conclusions differ significantly. 

We should also mention the work on random matrix theory by Potters and Bouchaud (\cite{PBRandMat}), which touches on many of the topics we consider in this paper. In particular their analysis of returns as products of Gaussians or t-distributions is very lose to our own. While many of the emphases are once again different to ours, we believe the general area of Random Matrix Theory to be a fruitful approach to trading strategies.

The primary tool we use to derive results is Isserlis' theorem \cite{Isserlis} or Wick's theorem (as it is known in the context of particle physics \cite{Wick}). This relates products and powers of multivariate normal random variables to their means and covariances. 
Wick's theorem has been applied in areas from particle physics, to quantum field theory to stock returns and there are some recent efforts to extend to non-Gaussian distributions (see, e.g., \cite{Michalowicz} for Gaussian-mixture and \cite{Kan} for products of quadratic forms and elliptic distributions), and it has been applied to continuous processes via the central limit theorem (see \cite{FCLTWick}). 
We have used these theorems in the context of dynamic (algorithmic) trading strategies to find expressions for the first four moments of strategy returns in closed-form. 
While it is not necessarily the aim of all scientific studies of trading strategies to find closed-form expressions, the ease with which we can describe strategy returns makes this direction relatively appealing and allows for a number of future extensions.

The paper is divided into sections on one asset, considered over a single period. With a normal signal, we will show there is a universal bound on the one-period Sharpe ratio, skewness and kurtosis. We explain the role of {\em total} or {\em orthogonal least squares} as an alternative to OLS for strategy optimisation.  
We look at the corresponding refinements to  measures of Sharpe ratio standard error for these dynamic strategies, improving on the large-sample theory based standard errors in more common use. We also introduce standard errors on skewness and kurtosis, which are distinct from those for Gaussian returns and present some basic results about multiple assets and diversification. Finally, we discuss the role of product measures, more pertinent to the study of dynamic strategies than simple Gaussian measures.
In the appendices, we present closed-form solutions to Sharpe ratios in the case of non-zero means. We also discuss extensions to our optimisations in the presence of transaction costs. We touch on the extension to multiple periods as well. As we mentioned, further extensions to over-fitting  by the use of covariance penalties (akin to Mallow's $C_p$ or AIC/BIC) have been presented separately in \cite{Firoozye2}.

\section{Single period linear strategies}

We consider the (log) returns of a single asset, ${R_{t}\sim\mathscr{N}(0,\sigma_R^{2})}$ returns with auto-covariance function at lag $k$, $\gamma(k)=E[R_t R_{t-k}]$, together with corresponding auto-correlation function (ACF), $c(k)=\gamma(k)/\gamma(0)$ at lag $k$. 

Our main aim is to work with strategies based on linear portfolio weights (or {\sl signals})
$X_{t}=\Sigma_{1}^{\infty}a_{k}R_{t-k}$ for coefficients $a_{k}$ generating the corresponding dynamic strategy returns $S_t=X_t\cdot R_t$ (here, and always, the signal, $X_t$ is assumed to only have appropriately lagged information).
Example strategy weights include exponentially weighted moving averages $a_{k}\propto\lambda^{k}$, simple moving averages $a_k = \frac{1}{T} \mathbbm{1}_{[1,\ldots,T]}$,  forecasts from ARMA models, etc. Most importantly, the  portfolio weights $X$ are normal and {\em jointly} normal with returns $R$. In Appendix \ref{Section:ConvolutionFilters}, we show that for a wide set of signals discussed in the Introduction, when applied to Gaussian returns, the signal and returns are jointly Gaussian.

We restrict our attention to return distributions over a single period. In the case of many momentum strategies, this period can be one day, if not longer. For higher-frequency intra-day strategies, this period can be much shorter. The pertinent concern is that the horizon (i.e., one period) is the same horizon over which the rebalancing of strategy weights is done. If weights are rebalanced every five minutes, then the single period should be five minutes. This is a necessary assumption in order to ensure the joint normality of (as yet indeterminate) signals and future returns. Moreover, this assumption will give some context to our results, which imply a maximal Sharpe ratio, maximal skewness and maximal kurtosis for dynamic linear strategies. 

We are interested in characterizing the moments of the strategy's unconditional returns, the corresponding standard errors on estimated quantities, and means of optimising various non-dimensional measures of returns such as the Sharpe ratio via the use of non-linear transformations of signals. Our goal is to look at unconditional properties of the strategy. It is important to avoid foresight in strategy design and this directly impacts the  conditional properties of strategies (e.g., conditional densities involve conditioning on the currently observed signal to determine properties of the returns, which are just Gaussian). In the context of our study, we are concerned with  one-period ahead returns of the   unconditional returns distribution of our strategy, where both the signals and the returns are unobserved, and the resulting distributions (in our case, the product of two normals) are much richer and more realistic -- for the interested reader, we have added a more detailed discussion of our framework in Appendix \ref{set-up}.

\subsection{Properties of linear strategies}

Given the joint normality of the signal and the returns, we can explicitly
characterise the one-period strategy returns (see \cite{Exact}). To allow for greater extendibility, we prefer to only consider the moments of the resulting distributions. These can be characterized easily using  Isserlis' theorem \cite{Isserlis}, which gives all moments for any multivariate normal random variable in terms of the mean and variance.  We also refer to  \cite{Haldane} who meticulously produces both non-central and central moments for powers and products of Gaussians. While this is a routine application of Isserlis' theorem, the algebra can be tedious, so we quote the results. 

\begin{theorem}[Isserlis (1918)]
If $X \sim \mathscr{N}(0,\Sigma)$,then 
$$E[X_1 X_2\cdots X_{2n}] = \sum_{i=1}^{2n} \prod_{i \neq j} E[X_i X_j]$$

and 
$$E[X_1 X_2\cdots X_{2n-1}] = 0$$
where the $\sum\prod $ is over all the $(2n)!/(2^n n!)$ unique partitions of $X_1,X_2,\ldots X_{2n}$ into pairs $X_i X_j$.
\end{theorem}

Haldane's paper quotes a large number of moment-based results for various powers of each normal. We quote the relevant results.

\begin{theorem}[Haldane (1942)]
\label{thm:Haldane}
If $x,y\sim \mathscr{N}(0,1)$ with correlation $\rho$ then 
\begin{eqnarray*}
E[xy] =& \rho\\
E[x^2y^2]=&1+2\rho^2\\
E[x^3y^3]=&3\rho(3+2\rho^2)\\
E[x^4y^4] =&3(3+24\rho^2+8\rho^4)
\end{eqnarray*}
and thus the central moments of $xy$ are
\begin{eqnarray}
\label{eqn:centered_moments1}
\mu_1 =& \rho\\
\label{eqn:centered_moments2}
\mu_2=&1+\rho^2 \\ 
\label{eqn:centered_moments3}
\mu_3=&2\rho(3+\rho^2)\\
\label{eqn:centered_moments4}
\mu_4= &3(3+14\rho^2+3\rho^4)
\end{eqnarray}
\end{theorem}

From these one period moments, (and a simple scaling argument giving the dependence on $\sigma(x)$ and $\sigma(y)$) we can characterise Sharpe ratio, skewness, etc., and can also define objective functions in order to determine some sense of optimality for a given strategy.

\begin{theorem}[Linear Gaussian]
\label{thm:Linear}
For single asset returns and a one period strategy, $\ensuremath{R_{t}\sim\mathscr{N}(0,\sigma_R^{2})}$  and $X_t\sim\mathscr{N}(0,\sigma_X^2)$ jointly normal with correlation $\rho$, the Sharpe ratio is given by
\be \SR = {\rho \over \sqrt{1+\rho^2}},
\label{Eqn:Sharpe}
\ee the skewness is given as
\be
\gamma_3=\frac{2\rho(3+\rho^{2})}{(1+\rho^{2})^{\frac{3}{2}}}, \label{Eqn:Skew}
\ee
and the kurtosis is given by
\be
\gamma_4 = \frac{3(3+14\rho^2+3\rho^4)}{(1+\rho^2)^2}
\label{Eqn:Kurt}
\ee
\end{theorem}
In the appendix, we extend  equations (\ref{Eqn:Sharpe}) and (\ref{Eqn:Skew})   to the case of non-zero means.

\begin{proof}
A simple application of  Theorem \ref{thm:Haldane} give us the following first two moments for our strategy $S_t = X_t \cdot R_t$:
$\mu_1=E[S_t]=E[X\cdot R] = \sigma_X\sigma_R \rho $.
and $\mu_2=Var[S_t]= \sigma_X^2\sigma_R^2(\rho^2+1)$
.
Thus we can derive the following results for the Sharpe ratio,
\begin{eqnarray*}
Sharpe=& {\mu_1\over \mu_2^{1/2}} \\
=& {\sigma_X\sigma_R\rho\over \sigma_X\sigma_R \sqrt{\rho^2+1}}\\
=& {\rho\over  \sqrt{\rho^2+1}}
\end{eqnarray*} 
Moreover, we can see that the skewness,
\begin{eqnarray*}
\gamma_3  =& {\mu_3\over \mu_2^{3/2}}\\
=& {2\rho (3+\rho^2)\over (1+\rho^2)^{3/2}}
\end{eqnarray*}
Finally, the kurtosis is given by
\begin{eqnarray*}
\gamma_4 =& {\mu_4\over \mu_2^{2}}\\
=&{3(3+14\rho^2+3\rho^4)\over (1+\rho^2)^2}
\end{eqnarray*}
\end{proof}

 If we restrict our attention to positive correlations, all three dimensionless statistics are monotonically increasing in $\rho$. Consequently, strategies that maximize one of these statistics will maximize the others, although the impact of correlation upon Sharpe ratio, skewness and kurtosis is different. We illustrate the cross-dependencies in the following charts, depicting the relationships between the variables. In figure \ref{fig:SharpeSkew}, the shaded blue histograms correspond to correlation ranges ($\{[-1,-0.5],[-0.5,0],[0,0.5],\ [0.5,1]\}$). We note that a uniform distribution in correlations maps into a higher likelihood of extreme Sharpe ratios and an even higher likelihood of extreme skewness and kurtosis. 

\begin{figure}[h]
	\centering
	\includegraphics[width=\linewidth]{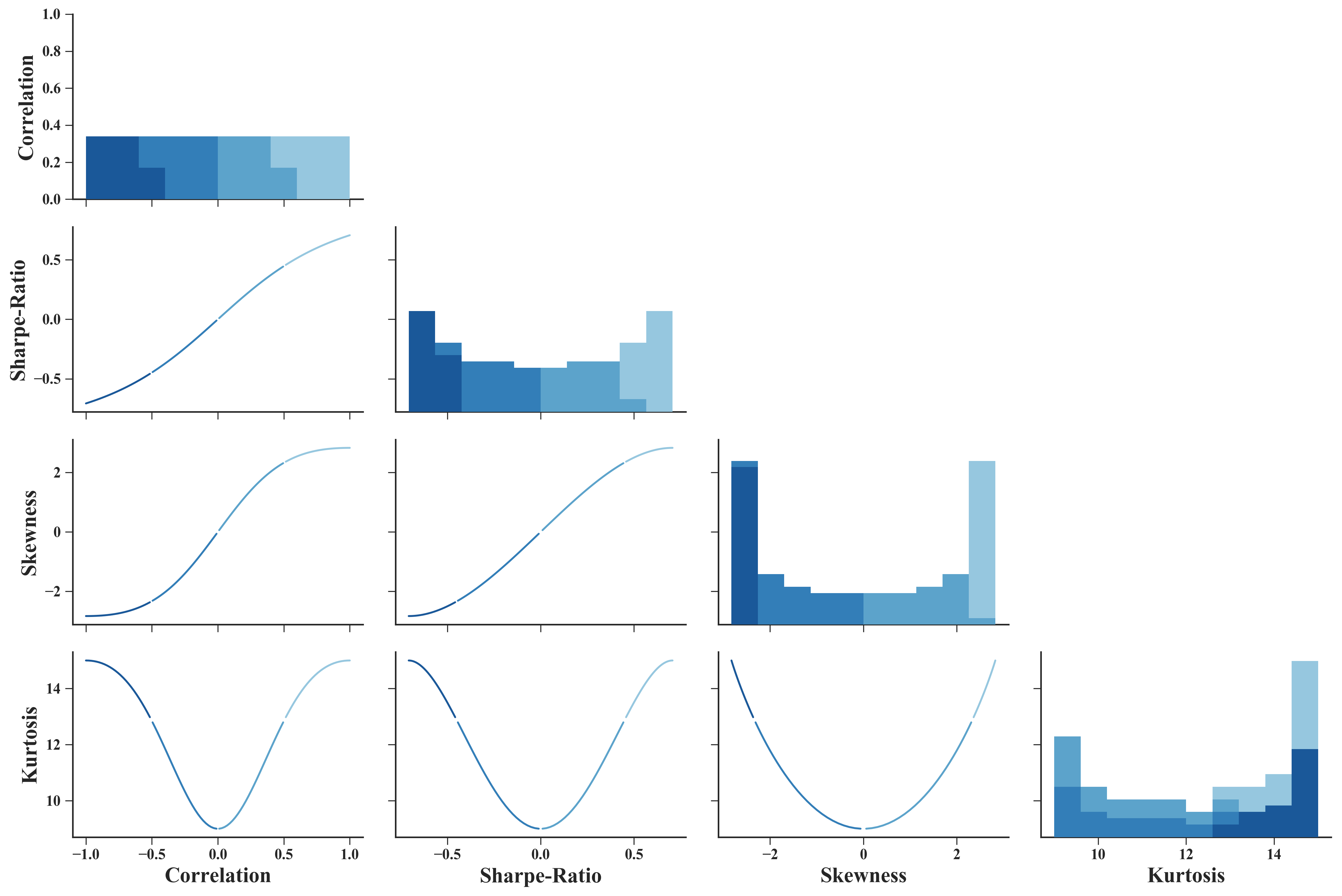}
	\caption{{\bf Correlation, Sharpe ratio, Skewness, and Kurtosis pairwise relationship.} A uniform distribution in correlation is bucketed into four ranges $\{[-1,-0.5],[-0.5,0],[0,0.5],\ [0.5,1]\}$ as depicted in the bar charts in shades of blue. After transforming the correlation into SR, $\gamma_3$ and $\gamma_4$ the frequencies are no longer uniform.} \label{fig:SharpeSkew}
\end{figure}

Skewness ranges in $[-2^{3/2},2^{3/2}]\approx[-2.8,2.8]$. Unlike the Sharpe ratio,
Skewness' dependence on correlation tends to flatten, so to achieve
90\% peak skewness, one needs only achieve a 0.60 correlation, while
for a 90\% peak Sharpe, one needs a correlation of 0.85. Kurtosis is an even function and varies from a minimal value of 9 to a maximum of 15. In practice, correlations will largely be close to zero and the resulting skewness and kurtosis significantly smaller than the maximal values.

Although we analyse  the moments of the strategy $S_t=X_t R_t$, the full product density is actually known in closed form (see appendix \ref{Section:FullDist},  \cite{Exact} and \cite{DistCorrel}). It is clear that the distribution of the strategy is {\sl leptokurtic} even when it is not predictive (when  the correlation is exactly zero, the strategy has a kurtosis of $9$). In the limit as $\rho\rightarrow 1$, the strategy's density approaches that of a non-central $\chi^2$, an effective {\sl best-case} density when considering the design of optimal linear dynamic strategies.

An optimised strategy with sufficient lags (and a means of ensuring parsimony) may be able to capture both mean-reversion and trend and result in yet higher correlations. Annualised Sharpe ratios of between 0.5-1.5 are most common (i.e., correlations of between 3\% to 9\%) for single asset strategies in this relatively low-frequency regime.


\subsection{Optimisation: Maximal Correlation, Total least squares}


Many algorithmic traders will explain how problematic  strategy optimisation is, given the  endless concerns of over-fitting, etc. Although these are a concern, the na\"{\i}ve use of strategies which are merely {\sl pulled out of thin air} is equally problematic, where there is no explicit use of optimisation (and, in its place more  {\em eye-balling} strategies or targeting Sharpe ratios rather loosely, effectively a somewhat loose mental optimisation exercise).  Practical considerations abound and real-world returns are neither Gaussian nor stationary. We argue irrespectively that using optimisation and a well-specified utility function as a starting point is a means of preventing strategies from being  just untested heuristics.  Unlike most discretionary traders' heuristics  (or {\em rules of thumb}) which have their 
place as a means of dealing with uncertainty (see for example \cite{Gigerenzer}), 
heuristic quantitative trading strategies run the risk of being entirely arbitrary, or are subject to a large number of human biases, in marked contrast to the monniker {\em quantitative} investment strategies.

\begin{figure}[h!]
\centering
	\includegraphics[width=\linewidth]{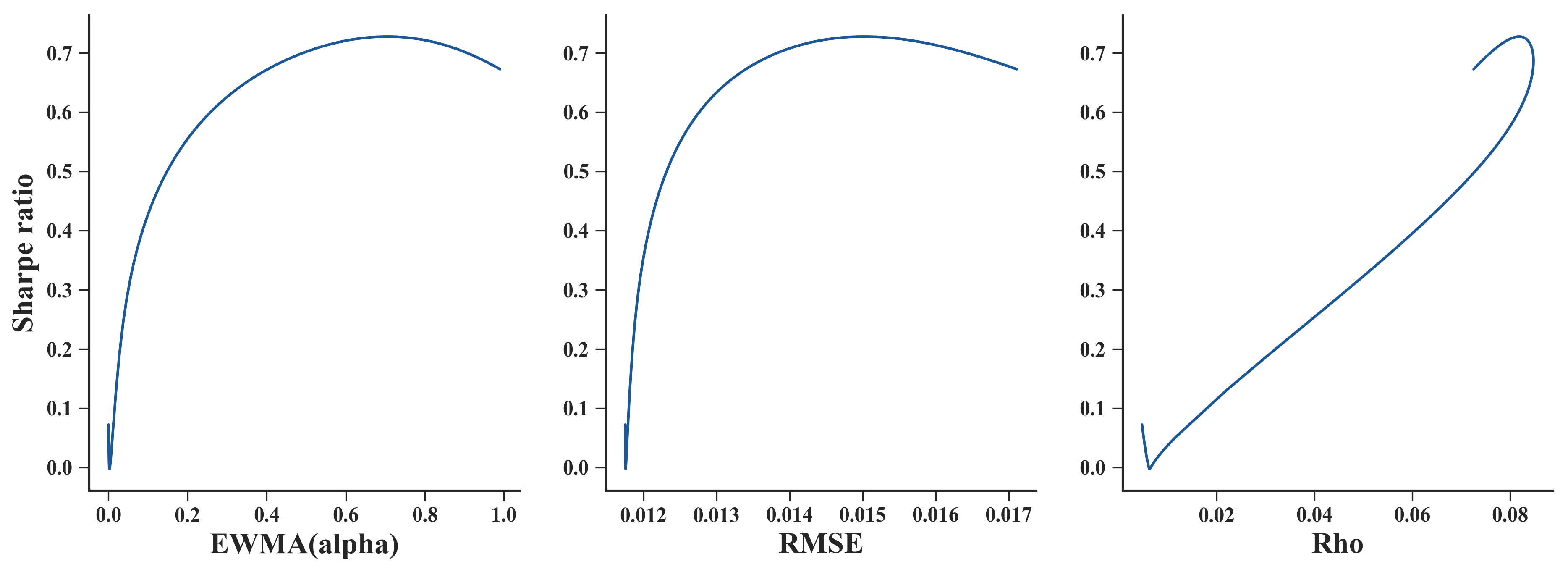}
  \caption{{\bf EWMA Strategy Sharpe Ratio vs $\alpha$, MSE and correlation} for S\&P 500 reversal strategies} \label{fig:EWMA_Sharpe_Corr}
\end{figure}

\begin{figure}[h!]
	\centering
	\includegraphics[width=\linewidth]{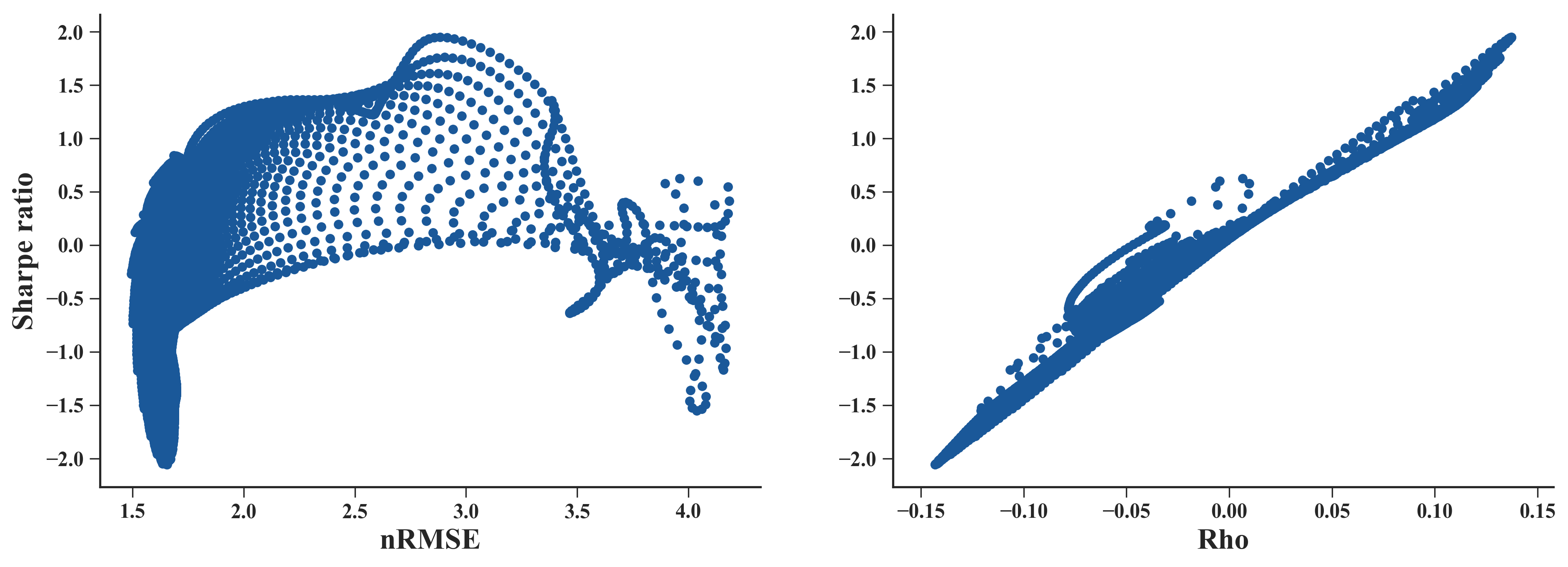}
	\caption{{\bf Holt-Winters Strategy Sharpe Ratio vs MSE and correlation} for S\&P 500 Reversal Strategies} \label{fig:HW_Sharpe_Corr}
\end{figure}

Where optimisation is used, the most common optimisation method is to minimize the mean-squared error (MSE) of the forecast. Our results show that  rather than to minimize the $\mathcal{L}^2$ norm between our signal and the forecast returns (or to maximize the likelihood), if the objective is to  maximize the Sharpe ratio, we must maximize the correlation. 

We can see in figures \ref{fig:EWMA_Sharpe_Corr} and \ref{fig:HW_Sharpe_Corr}, a depiction of fits of strategies applied to S\&P 500 using EWMA and HW filters for a variety of parameters. The relationship between MSE and Sharpe ratio is not monotone in MSE for the EWMA filter as we see in figure \ref{fig:EWMA_Sharpe_Corr}, while it is much closer to being linear in the case of the relationship between correlation and Sharpe. For the case of HW (with two parameters), in figure  \ref{fig:HW_Sharpe_Corr} any given MSE can lead to a non-unique Sharpe ratio, sometimes with a very broad range, leading us to conclude that the optimization is poorly posed. The relationship of correlation to Sharpe is obviously closer to being linear, with higher correlations almost always leading to higher Sharpe ratios.

In the case of a one-dimensional forecasting problem with (unconstrained) linear signals, optimizing the correlation amounts to using  what is known as {\sl total least squares regression (TLS)} or {\sl orthogonal distance regression}, a form of principal components regression (see, e.g., \cite{Golub} and \cite{Markovsky}).  In the multivariate case, it would be more closely related to {\sl canonical correlation analysis} (CCA). 

Unlike OLS, where the dependent variable is assumed to be measured with error and the independent variables are assumed to be measured without error, in total least squares regression, both dependent and independent variables are assumed to be measured with error, and the objective function compensates for this by minimizing the sum squared of orthogonal distances to the fitted hyperplane. This is a simple form of errors-in-variables (EIV) regression and has been studied since the late 1870s, and is most closely related to principal components analysis. For $k$ regressors, the TLS fit will produce weights which are orthogonal to the first $k-1$ principal components. 

So, if we consider the signal $X=Z\beta$ to be a linear combination of features, with $Z\in \bold{R}^k$ a $k$-dimensional feature space, then we note that
$$\hat{\beta}^{OLS} = (Z'Z)^{-1}Z'R$$ but
$$\hat{\beta}^{TLS} = (Z'Z-\sigma_{k+1}^2 I)^{-1}Z'R$$ where $\sigma_{k+1}$ is the smallest singular value for the $T\times (k+1)$ dimensional matrix $\tilde{X}=[R,Z]$ (i.e., the concatenation of the features and the returns, see, e.g., \cite{Rahman-Yu}\footnote[4]{A more common method for extracting TLS estimates is via a PCA of the concatenation matrix $\tilde{X}$, where $\hat{\beta}^{TLS}$ is chosen to cancel the least significant principal component}).It is well known that, for the case of OLS, the smooth or hat matrix $\hat{R} = M R$ is given by
$$M^{OLS} = Z(Z'Z)^{-1} Z'$$
with $\tr(M^{OLS})=k$, the number of features. 
In contrast, 
$$M^{TLS}=Z(Z'Z-\sigma_{k+1}^2 I)^{-1}Z'$$
and  effectively has a greater number of degrees of freedom than that of OLS, i.e., $$\tr(M^{TLS})\geq \tr(M^{OLS})$$ with equality only when there is complete collinearity\footnote[5]{In this case, it is also known that $\tr(M)=\tr(L)$ where $L= (Z'Z-\sigma_{k+1}^2 I)^{-1}Z'Z$ and we know that the singular values of $\sigma(L)=\{{\lambda_i^2}/{(\lambda_i^2-\sigma_{k+1}^2)}\}$ where $\lambda_i$ are the singular values of $Z$ (or correspondingly, $\lambda_i^2$ are the singular values of $Z'Z$), and $\lambda_1\geq\cdots\geq\lambda_k>0$ (\cite{Leyang}).  By the Wilkinson interlacing theorem, $\lambda_k\geq\sigma_{k+1}\geq0$ (see \cite{Rahman-Yu}). Consequently,
$$\tr(M^{TLS})=\sum_i\frac{\lambda_i^2}{(\lambda_i^2-\sigma_{k+1}^2)}\geq k=\tr(M^{OLS})$$ 
with equality iff $\sigma_{k+1}^2=0$ (i.e., when there the $R^2=100\%$ and consequently, OLS and TLS coincide). In other words,  $\tr(M^{TLS})\geq\tr(M^{OLS})$.  
}
For this reason, many people see TLS as  an {\em anti-regularisation} method and may result in less-stable response to outliers (see for example, \cite{Zhang}).  Consequently, there is extensive study of {\em regularised} TLS, typically using a weighted ridge-regression (or Tikhonov) penalty (see discussion in \cite{Zhang} for more detail on this large body of research). The stability of TLS in out-of-sample performance is an issue we broach in our study of over-fitting penalties (see \cite{Firoozye2}).


While maximizing correlation rather than minimizing the MSE seems a very minor change in objective function, the formulas differ from those of standard OLS. The end result is a linear fit which takes into account the errors in the underlying conditioning information. We believe that it should be of relatively little  consequence when the features are appropriately normalized, as is the case for univariate time-series estimation, although some authors have suggested that optimising TLS is not appropriate for prediction (see, e.g., \cite{Fuller} section 1.6.3).  When we seek to maximize the Sharpe ratio of a strategy, the objective should {\em not be} prediction, but rather optimal weight choice.

%

\subsection{Maximal Sharpe ratios, Maximal Skewness, Minimal Kurtosis}

Surprisingly, there appears to be a maximal Sharpe ratio for linear strategies. In the case of normal signals and normal returns, the maximal Sharpe ratio is that of a non-central $\chi^2$ distribution and the resulting maximal statistics are
\begin{eqnarray*}
\SR^{max}=&\frac{\sqrt{2}}{2}\approx 0.707\\
\gamma_3^{max} = &2\sqrt{2}\approx 2.828\\
\gamma_4^{max} =& 15.000
\end{eqnarray*}

While  the estimate for the Sharpe ratio may seem surprisingly low, we comment that these are for a single period, for one  single rebalancing. For a daily rebalanced strategy, if we na{\"i}vely annualize the Sharpe ratio (by a factor of $\sqrt{252}$), we get a maximal Sharpe of approximately $SR_{max}\approx 11.225$, a level generally well beyond what is attained in practice. The statistics, $\gamma^{max}_3$ and $\gamma^{max}_4$ do not scale when annualized, but are still  large irrespective of the time horizon.

We note that our assumption of normality could easily be relaxed by considering non-linear transforms of the signals $X$ with the end-result that the maximal Sharpe Ratio bounds are relaxed. While this is beyond the scope of the current paper, we note that it is easy to show that simple non-linear strategies, going long one unit if the signal is above a threshold $k$ and short one unit if it below $-k$, i.e., $f_k(X)=\mathbbm{1}_{X>k}-\mathbbm{1}_{X<k}$ can be shown to have arbitrarily large Sharpe Ratios, depending on the choice of threshold, $k$. The probability of initiating such an arbitrarily high Sharpe ratio trade  likewise decreases to being negligible. Thus, stationary returns with a small non-zero autocorrelation can lead to violations of Hansen-Jagannathan (or {\em good deal} bounds).

Noticeable as well from these formulas is that, while Sharpe and skewness may change sign, kurtosis is always bounded below and takes a minimum value of $9$ (i.e., an excess kurtosis of $6$). Normality of the resulting strategy returns is not a good underlying assumption, since the theoretical value of the Jarque-Bera test would be, at 
\begin{eqnarray*}
JB(n) &=& \frac{n-k+1}{6}(\gamma_3^2 + \frac{(\gamma_4-3)^2}{6})\\
 &\geq& \frac{(n-k+1)}{6}(\frac{36}{4})\\
 & =& 1.5(n-k+1)
 \end{eqnarray*}
and this is asymptotically $\chi^2(2)$ (i.e., rejection of normality at a 0.99 confidence interval of $JB>9.210$). Theoretically, we would need a relatively small sample to be able to reject normality.

\section{Refined Standard Errors}

\label{section:Stderr}

Given that we have closed-form estimates of a number of relevant statistics for dynamic linear strategies, it makes sense to consider the effects of estimation error upon quantities such as the Sharpe ratio. Many analysts and traders who consider dynamic strategies in practice will consider altering the strategies on an ongoing basis, and are typically in a quandary over whether the observed change in Sharpe ratio or skewness, when they make changes to their strategies, are in fact statistically significant.

\subsection{Standard Errors for Sharpe Ratios}
While there are formulas for standard errors for  Sharpe ratios of generic assets, these are not specific to Sharpe ratios generated by dynamic trading strategies, and as a consequence, there is some possibility of refining them.

We refer to \cite{Pav} for an exhaustive overview of the mechanics of Sharpe ratios, and in particular, Section 1.4, quoting many of the known results about standard errors. Specifically, we look to \cite{Lo} for large-sample estimates of standard errors for Sharpe ratios of generic assets, given the asymptotic normality of returns.  For a sample of size $N$ and IID returns, he obtains the large-sample distribution,
$$\widehat{\SR} \sim \mathscr{N}\left(\SR,\stderr_{\Lo}^2\right),$$ so a standard error, $\stderr_{\Lo} = \sqrt{(1+\frac{1}{2} \SR^2)/T}$ which he suggests should be approximated using standard error $\sqrt{(1+\frac{1}{2} \widehat{\SR}^2)/T}$.

While Lo's estimates may be appropriate for generic assets,  for Sharpe ratios derived from dynamic strategies, we have a somewhat more refined characterisation of the variability of the estimated Sharpe ratios. With correlated Gaussian signals and returns, we derive the following result

\begin{corollary}[Stderrs] For returns $R_t\sim \mathscr{N}(0,\sigma_R^2)$ and signal 
$X_t\sim \mathscr{N}(0,\sigma_X^2)$ with correlation $\rho$, and sample size $T$, the standard errors are given by 
\begin{eqnarray}
\label{eqn:stderr_sharpe}
\stderr_{\implied} &=\frac{1}{(\hat{\rho}^2+1)^{3/2}}\sqrt{\frac{1-\hat{\rho}^2}{T-2}}\\
\label{eqn:stderr_sharpe2}
&\approx (1-\widehat{\SR}^2)\sqrt{\frac{1-2\widehat{\SR}^2}{T-2}}
\end{eqnarray}
for $|\widehat{\SR}|<\sqrt{2}/2$.

\end{corollary}

\begin{proof}
As is well known, for a bivariate Gaussian process of sample size $T$, the distribution for the sample ({\em Pearson}) correlation is given by
\be\hat{\rho} \sim f_\rho(\hat{\rho})=
\frac{(T-2)(1-\rho^2)^{(T-1)/2}(1-\hat{\rho}^2)^{(T-4)/2}}{\pi}
\int_0^\infty\frac{dw}{(\cosh(w)-\rho\hat{\rho})^{T-1}}\label{Eqn:rhohat}\ee

The  standard errors which approximate those in equation (\ref{Eqn:rhohat}) for $\hat{\rho}$ are 
$$\stderr_{\rho}=\sqrt{\frac{1-\hat{\rho}^2}{T-2}}$$ 
(attributed to Sheppard, and used by Pearson, see, e.g., \cite{Hald}). Taken together with the results of Theorem \ref{thm:Linear}, we apply the delta method to find that the resulting standard errors for our plug-in estimate for the Sharpe ratio,
$\widehat{\SR}=\frac{\hat{\rho}}{\sqrt{\hat{\rho}^2+1}}$ is given by

\begin{eqnarray*}
\stderr_{\implied}&=&
\frac{\partial \widehat{\SR}}{\partial \hat{\rho}}\cdot \stderr_\rho\\
&=&\frac{1}{(\hat{\rho}^2+1)^{3/2}}\sqrt{\frac{1-\hat{\rho}^2}{T-2}}.
\end{eqnarray*}
which gives us equation (\ref{eqn:stderr_sharpe}). If we solve for $\hat{\rho}$ in terms of $\widehat{\SR}$, we are able to derive equation (\ref{eqn:stderr_sharpe2}).
\end{proof}

We note that in spite of the fact that Lo's standard errors are very near our estimates for large sample size, the entire sampling distribution from our estimates are much more concentrated than the $\mathscr{N}(0,\stderr_{\Lo}^2)$, potentially leading to tighter confidence intervals at the 99\% or higher confidence levels. We can see that the tail of the distribution given by Lo is much fatter than ours, in figure (\ref{fig:ViolinPlot}).

\begin{figure}[h!]
	\centering
	\subfloat{\includegraphics[width=0.33\linewidth]{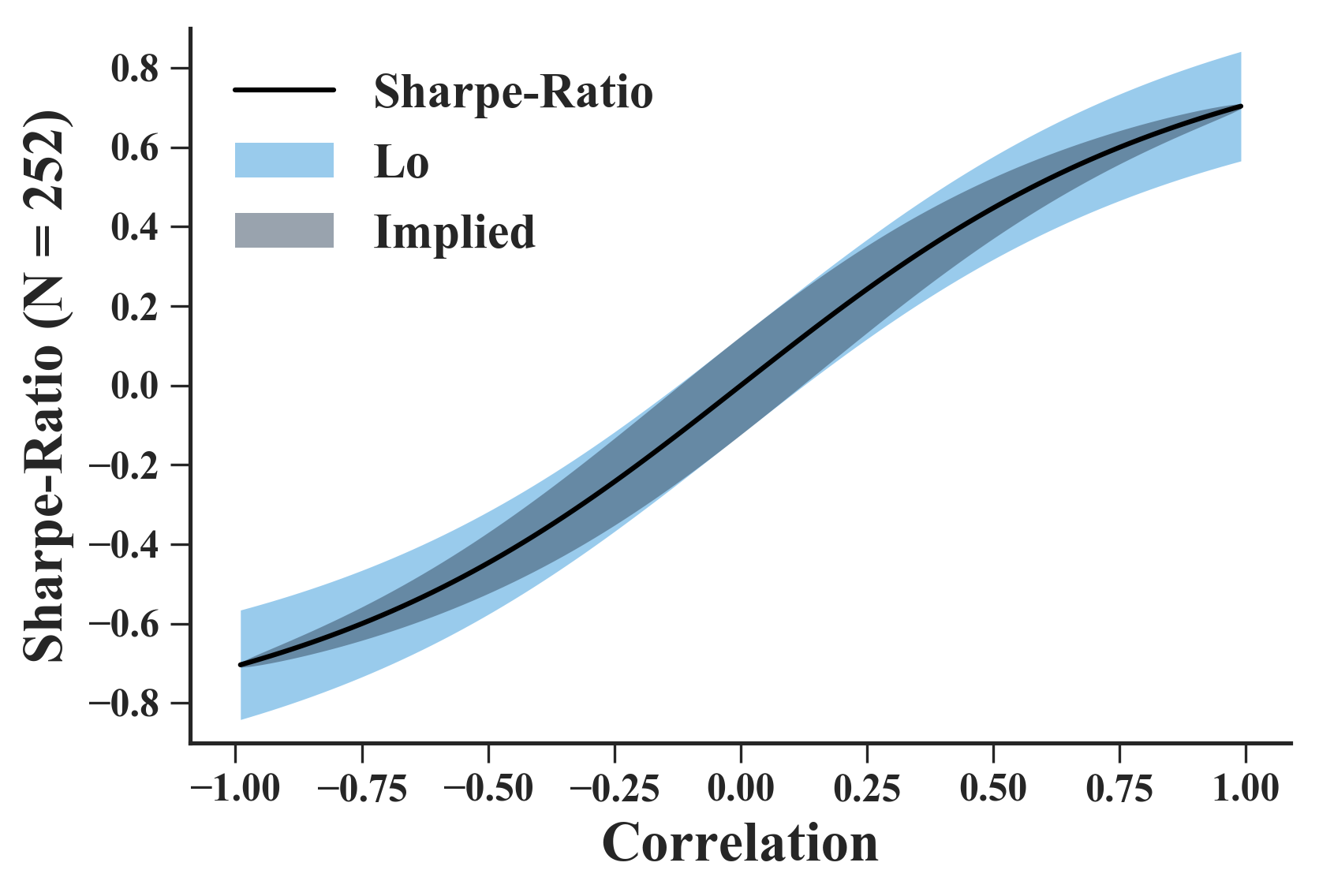}}
	\subfloat{\includegraphics[width=0.33\linewidth]{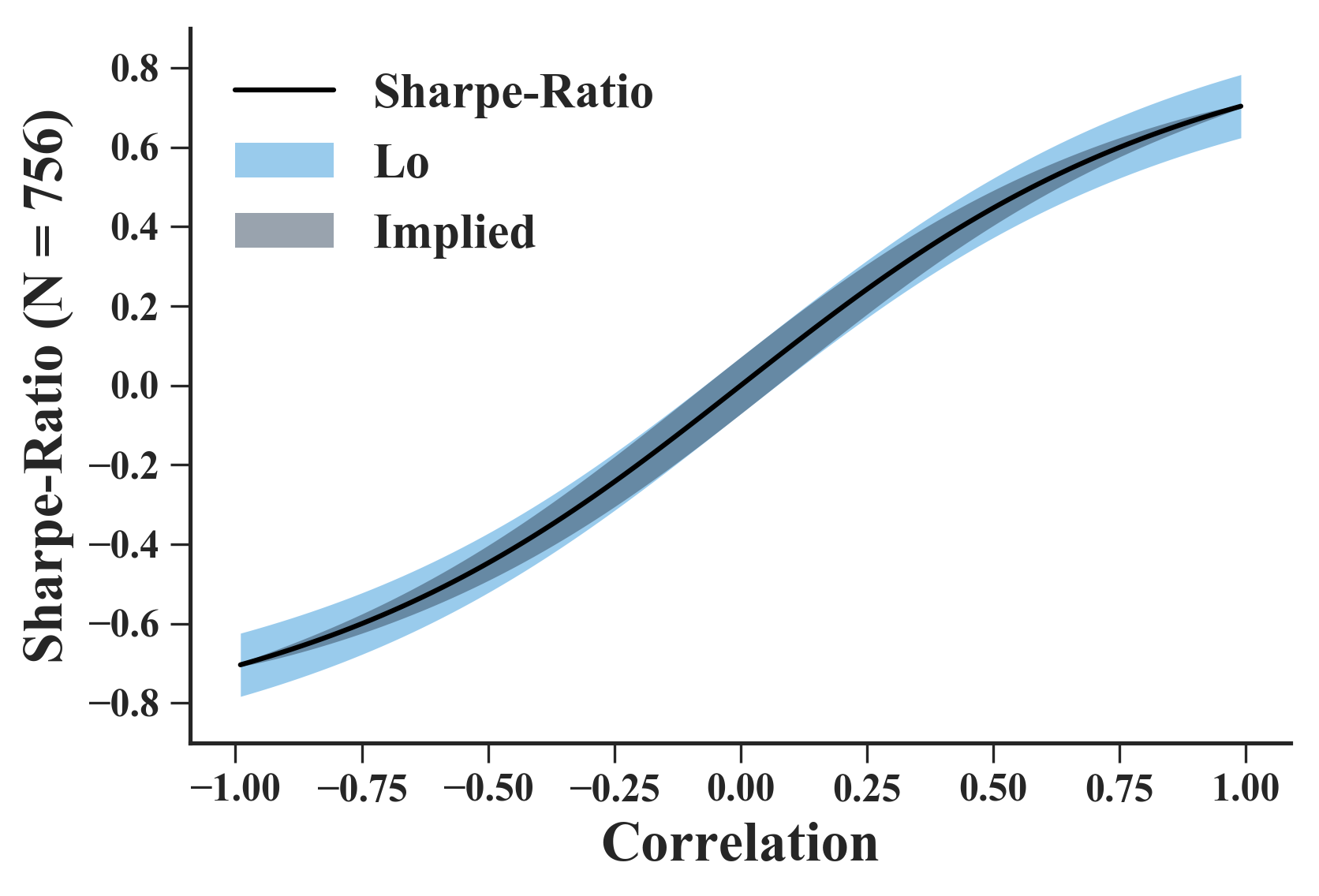}}
	\subfloat{\includegraphics[width=0.33\linewidth]{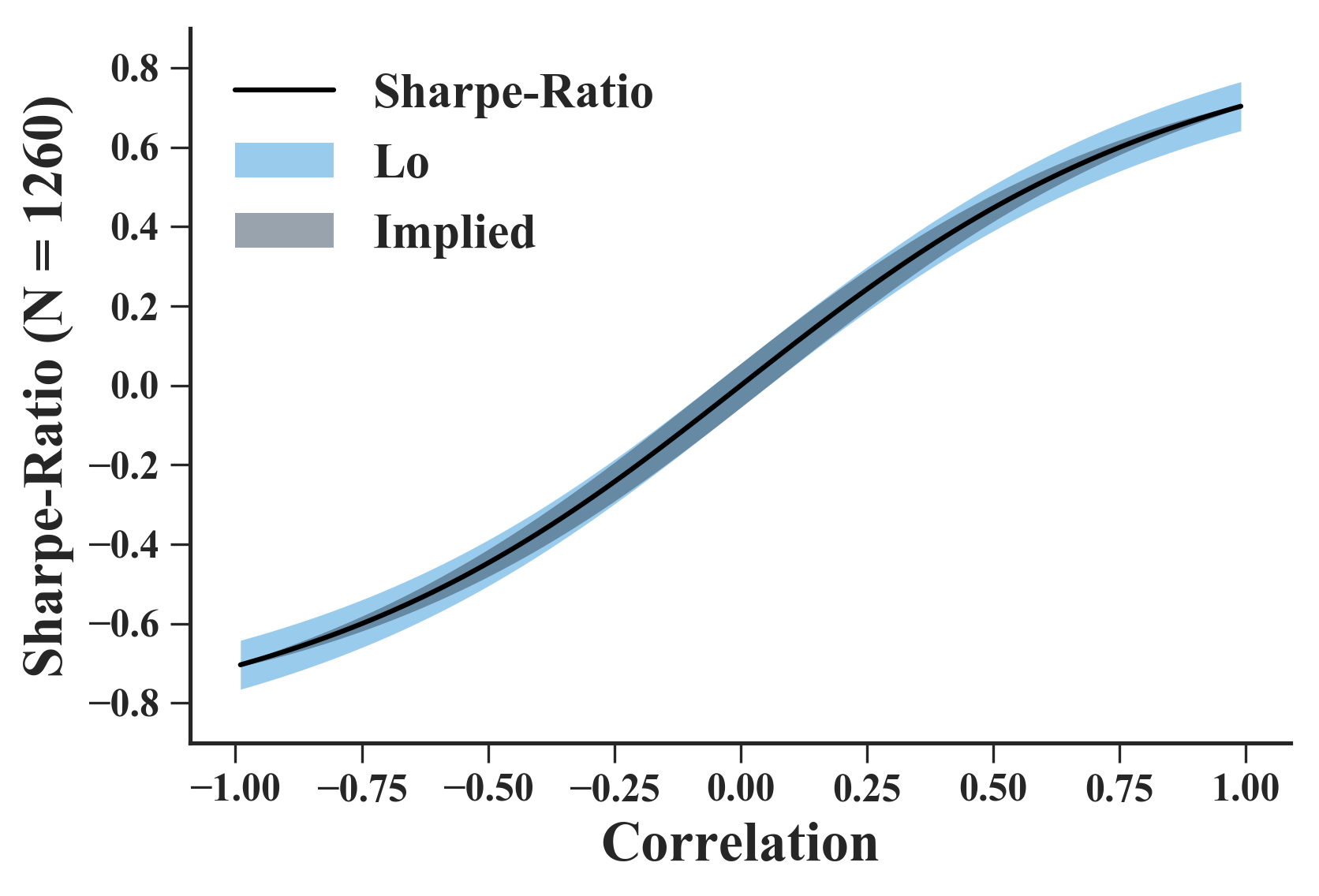}}
	\caption{{\bf Sharpe ratio and Confidence Interval Comparisons, based on different sample sizes.} We note that the {\em implied} confidence intervals are within Lo's, although primarily for larger predictive power.} \label{Lo-Us} \label{fig:CI1yr}
\end{figure}

\begin{figure}[h!]
	\centering
	\includegraphics[width=\textwidth]{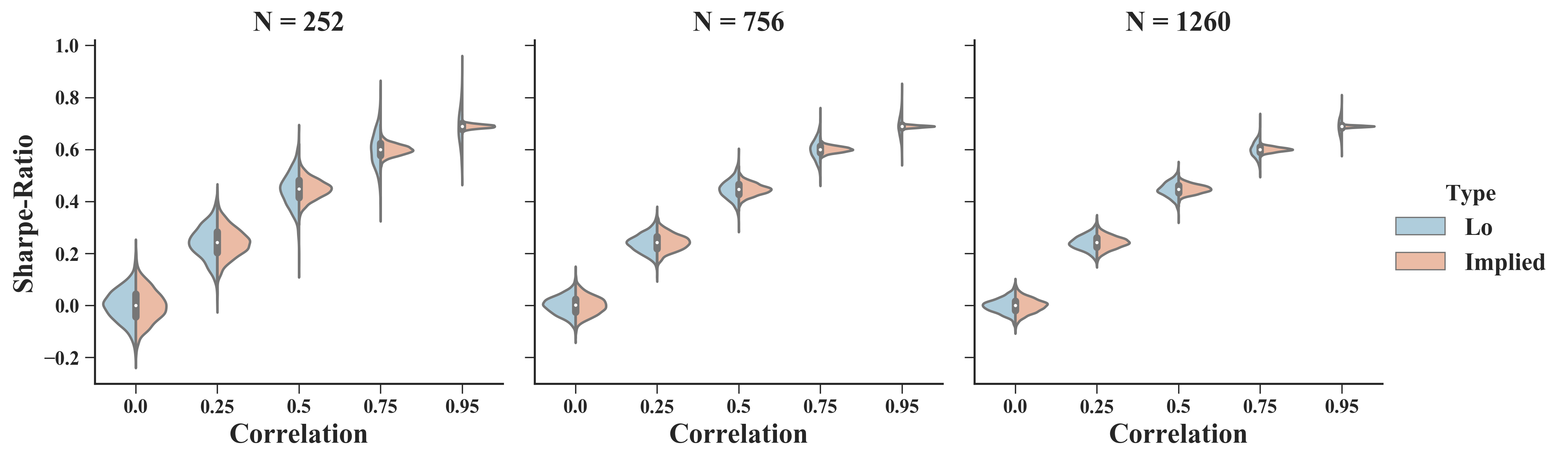}
	\caption{{\bf Sharpe ratios full distribution} While the $95^{th}$ percentile shows close agreement between Lo's large-sample standard errors and {\em implied} standard errors, the distribution of {\em implied} is far more fat-tailed.}
	\label{fig:ViolinPlot}
\end{figure}

\begin{figure}
	\centering
	\subfloat{\includegraphics[width=0.33\linewidth]{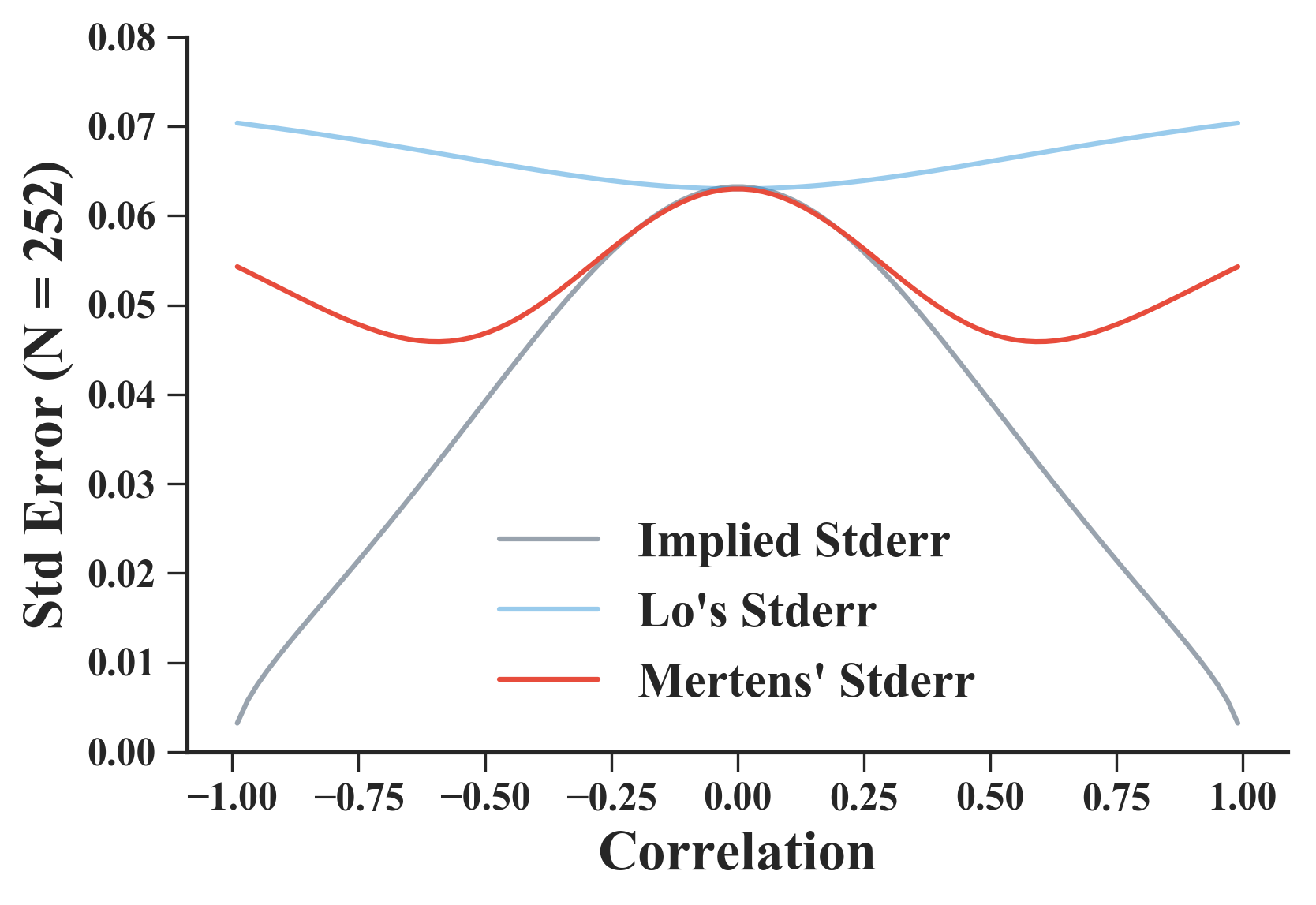}}
	\subfloat{\includegraphics[width=0.33\linewidth]{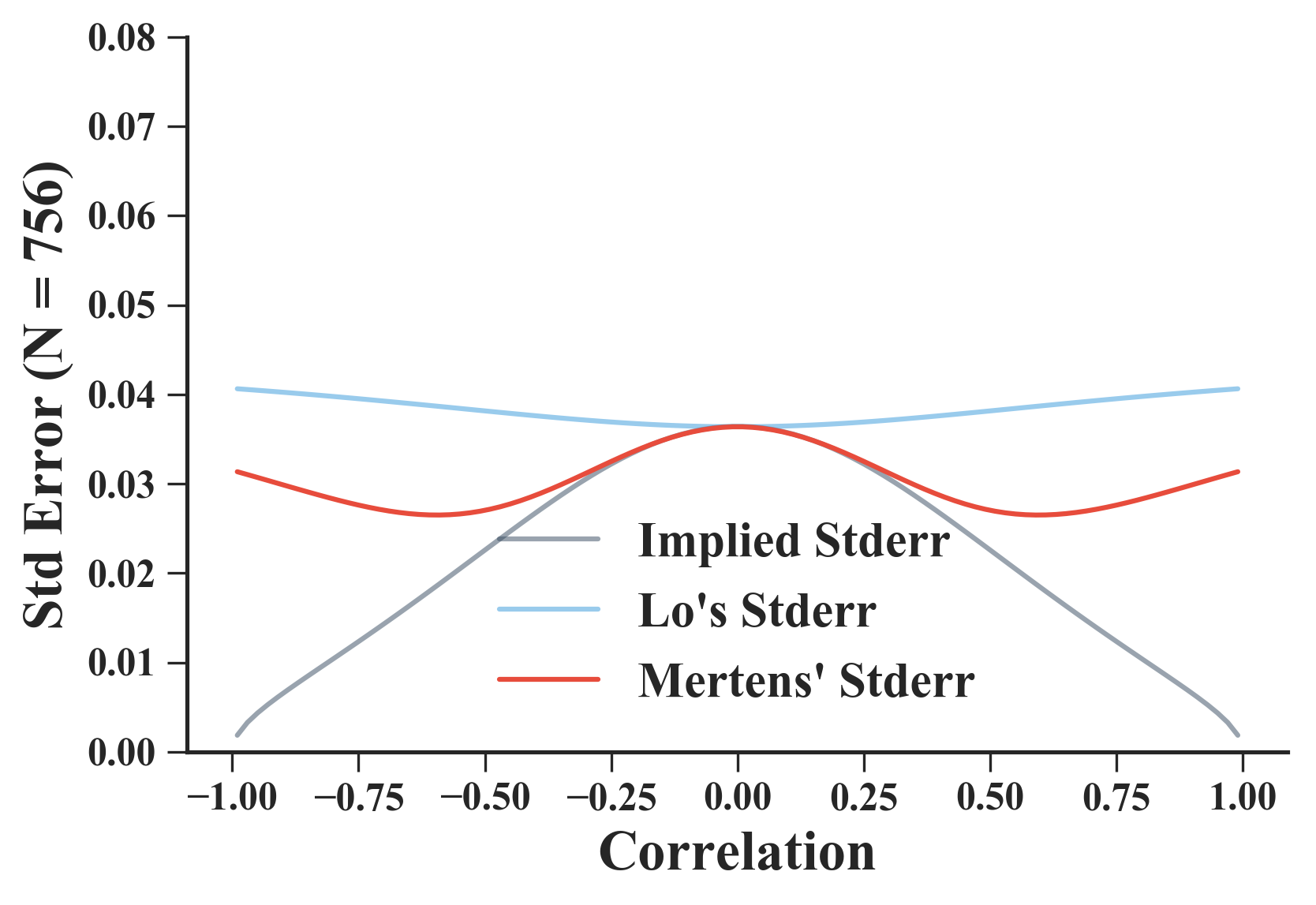}}
	\subfloat{\includegraphics[width=0.33\linewidth]{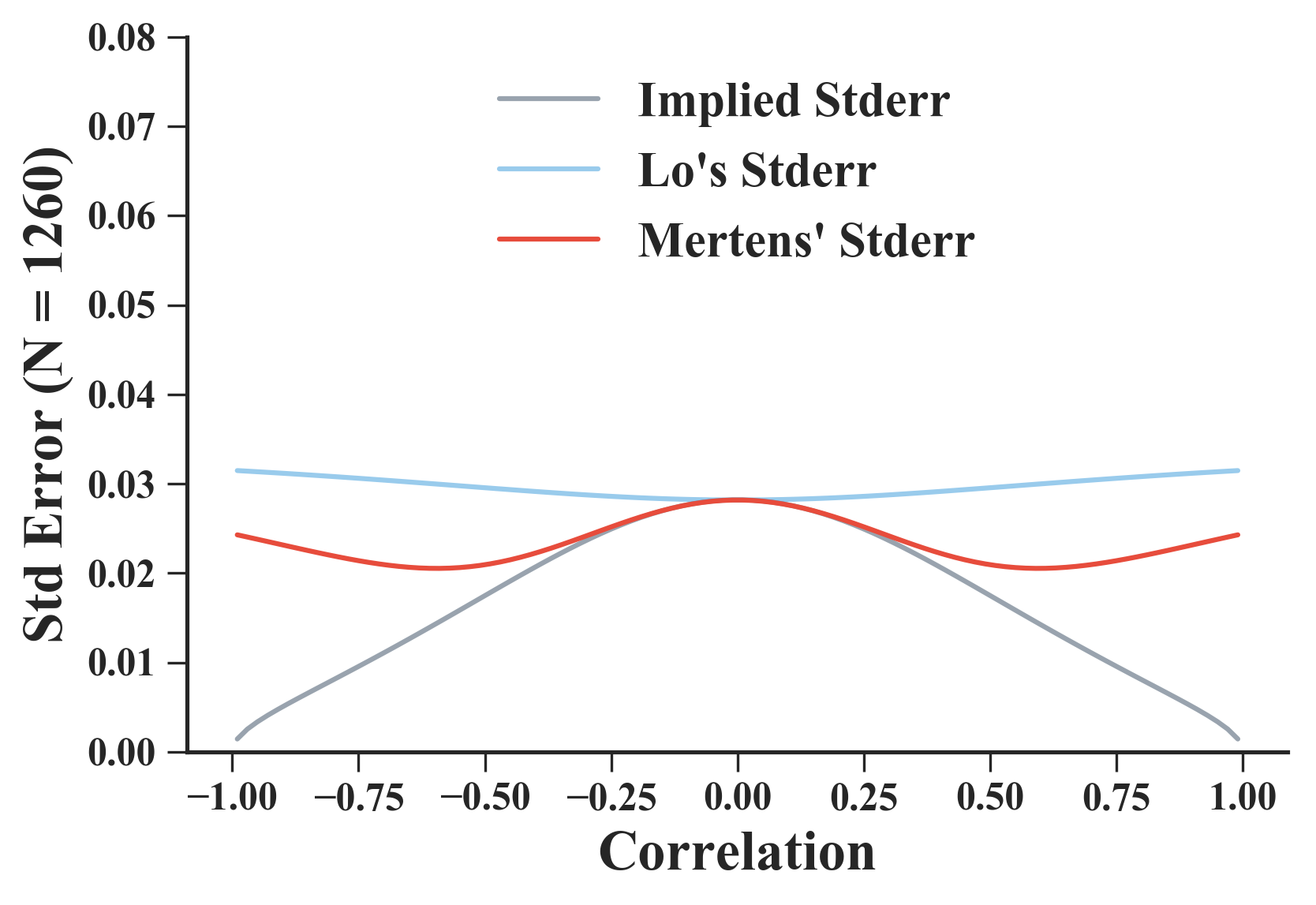}}
	\caption{{\bf Standard errors based on different sample sizes and formulas.} Ignoring parameter uncertainty, Merten's adjustment to Lo's standard errors improves standard errors to be nearly as tight as {\em implied}. In practice, parameter uncertainty hurts the performance.
		\label{fig:StdErrMertensLoImplied}}
\end{figure}

Mertens gives a refinement of Lo's result (\cite{Mertens})  by including adjustments for skewness and excess kurtosis:
\be \stderr^2_{\Mertens} = \bigl(1+\frac{1}{2}\hat{\SR}^2-\gamma_3\cdot\hat{\SR}+\frac{\gamma_4-3}{4}\cdot\hat{\SR}^2\bigr).
\label{Eqn:Mertens} \ee
If we use our plug-in estimates for skewness and excess kurtosis (i.e., coming from equations (\ref{Eqn:Skew} and \ref{Eqn:Kurt})) into equation (\ref{Eqn:Mertens}) we are able to find a modestly tighter estimate of the standard error than Lo. For most smaller amplitude correlations, this estimate comes very close to our estimate of standard error (see figure (\ref{fig:StdErrMertensLoImplied})) and for small $N$ and low correlations, Lo's standard errors are in fact tighter.  For large correlations, our standard errors are significantly tighter. For large sample sizes, there is little difference between them. Using our estimates for $\gamma_3$ and $\gamma_4$, Mertens' approximation is always tighter than Lo's; in particular for correlations $|\rho|<0.5$, Mertens' approximation appears almost identical to our own. Irrespective, we argue in section \ref{Section:Gauss vs Product} that our standard errors are more appropriate for dynamic strategies if there is any significant difference between the measures.

\subsection{Standard Errors for Higher Moments}
Using exactly the same procedure, we can easily derive standard errors for both skewness and kurtosis. In terms of classical confidence intervals, we consider \cite{Joanes} and \cite{Cramer}   which apply to Gaussian (and non-Gaussian distributions), noting that   \cite{Lo} is a broader result on the large-sample limits of Sharpe Ratios. We are concerned  with Pearson skewness and  kurtosis, i.e., 
\begin{eqnarray*}
\gamma_3 &= \frac{\mu_3}mu_2^{3/2}\\
\gamma_4 &= \frac{\mu_4}{\mu_2^2}
\end{eqnarray*}
although it is not hard to consider other definitions of skewness and kurtosis using unbiased estimators of the moments as are given in \cite{Joanes}, in this case originally from \cite{Cramer}.
Given these definitions, under the assumption of normality for the underlying returns (or correspondingly, using large-sample limits) where the sample size is $T$, standard errors are given as
\begin{eqnarray*}
\stderr_{\gamma_3} &= \sqrt{\frac{6(T-2)}{(T+1)(T+3)}}\\
\stderr_{\gamma_4} &= \sqrt{\frac{24 T (T-2)(T-3)}{(T+1)^2(T+3)(T+5)}}
\end{eqnarray*}

In the case of dynamic strategies, using our assumption of normal signal and normal returns, we are able to derive the following:

\begin{corollary}[Higher moment standard errors]
\label{higherStderr} 
For returns $R_t\sim \mathscr{N}(0,\sigma_R^2)$ and signal 
$X_t\sim \mathscr{N}(0,\sigma_X^2)$ with correlation $\rho$, and sample size $T$, the standard errors are given by\footnote[6]{ While $\rho$ can be expressed in terms of either $\gamma_3$ or $\gamma_4$ to eliminate $\rho$ from these expressions, unlike the case of the standard errors of the Sharpe ratio, the expressions are too complicated to be that useful.}
 
\begin{eqnarray*}
\stderr_{\gamma_3} &=-\frac{6(\hat{\rho}^2-1)}{(\hat{\rho}^2+1)^{5/2}}\cdot\sqrt{\frac{1-\hat{\rho}^2}{T-2}}
\end{eqnarray*}
and
\begin{eqnarray*}
\stderr_{\gamma_4} &=-\frac{48\hat{\rho}(\hat{\rho}^2-1)}{(\hat{\rho}^2+1)^{3}}\cdot\sqrt{\frac{1-\hat{\rho}^2}{T-2}}
\end{eqnarray*}
for $|\hat{\rho}|<1$.

\end{corollary}

We rely on the delta-method, recognizing that 
$\stderr_{\gamma_k} =\partial{\gamma_k}/\partial{\rho}\cdot \stderr_{\rho}$ for $k=3,4$. Given the following easily calculated derivatives:
\begin{eqnarray}
\label{Eqn:SkewKurtStderr}
\frac{\partial \gamma_3}{\partial \rho} & = -\frac{6(\rho^2-1)}{(\rho^2+1)^{5/2}}\\
\frac{\partial \gamma_4}{\partial \rho} & = -\frac{48\rho(\rho^2-1)}{(\rho^2+1)^{3}}
\end{eqnarray}

As we can tell from the formulas in corollary (\ref{higherStderr}), the derived standard errors for both skewness and kurtosis collapse to zero when $\rho=1$.

\begin{figure}
	\centering
	\subfloat{\includegraphics[width=0.33\linewidth]{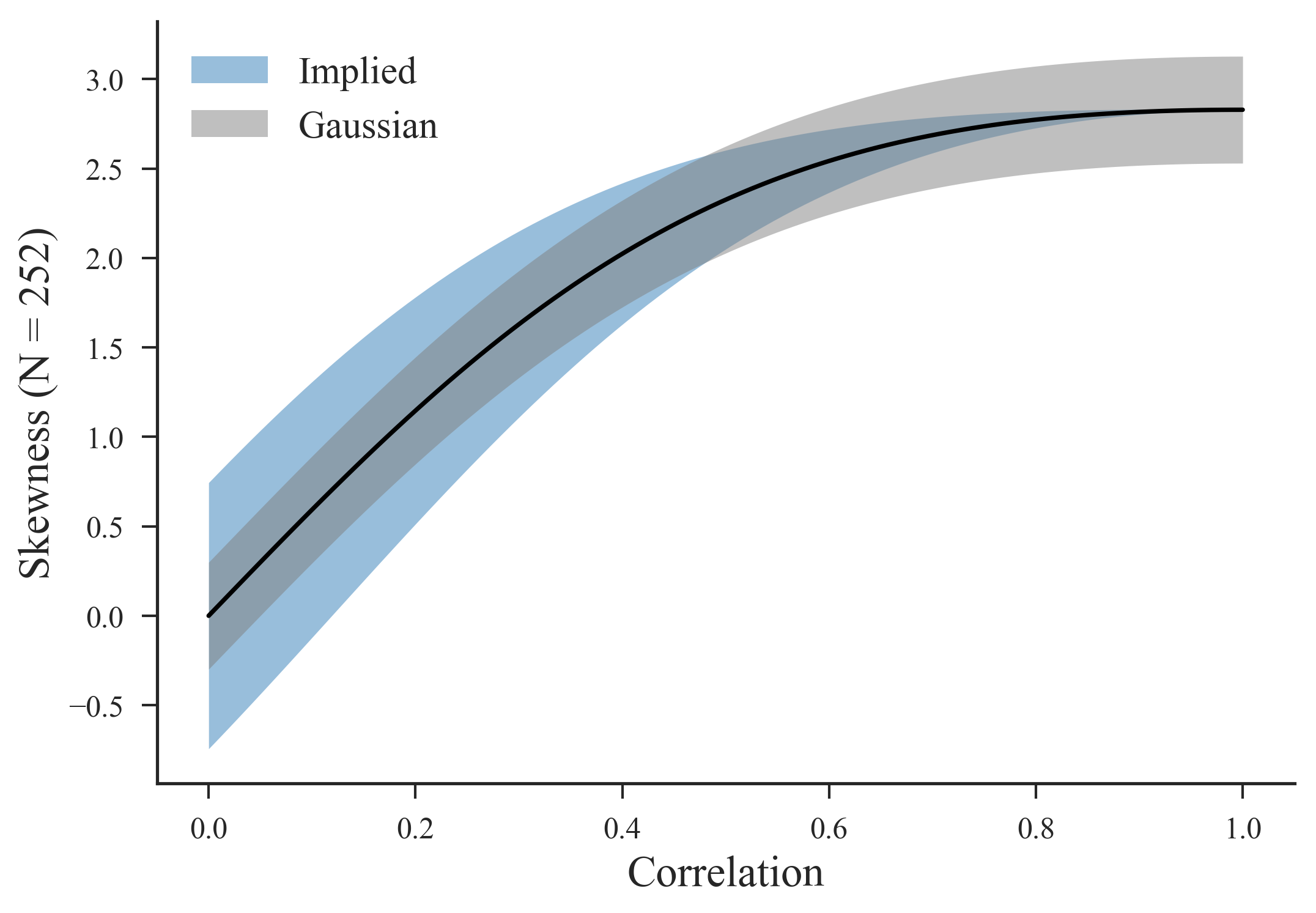}}
	\subfloat{\includegraphics[width=0.33\linewidth]{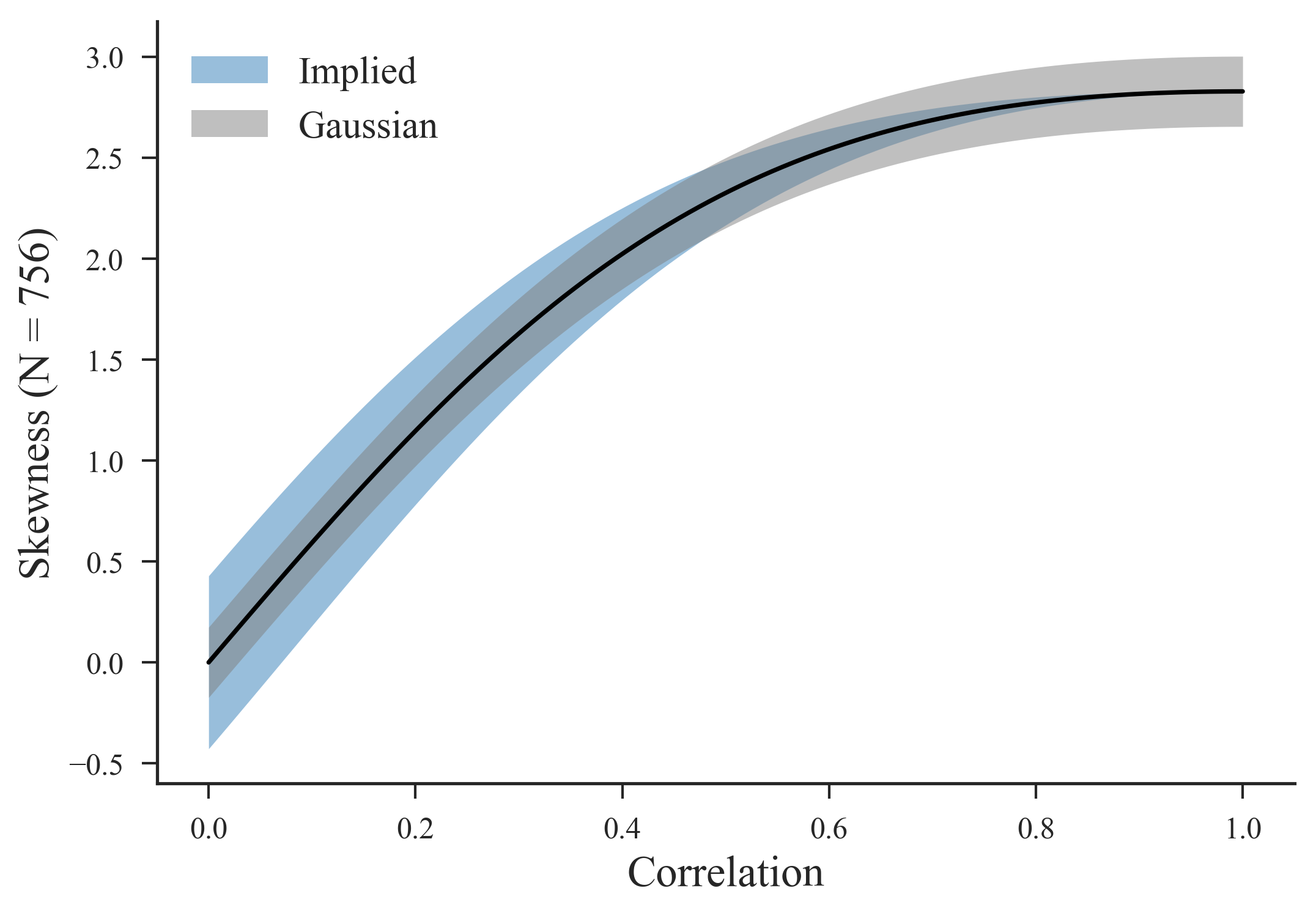}}
	\subfloat{\includegraphics[width=0.33\linewidth]{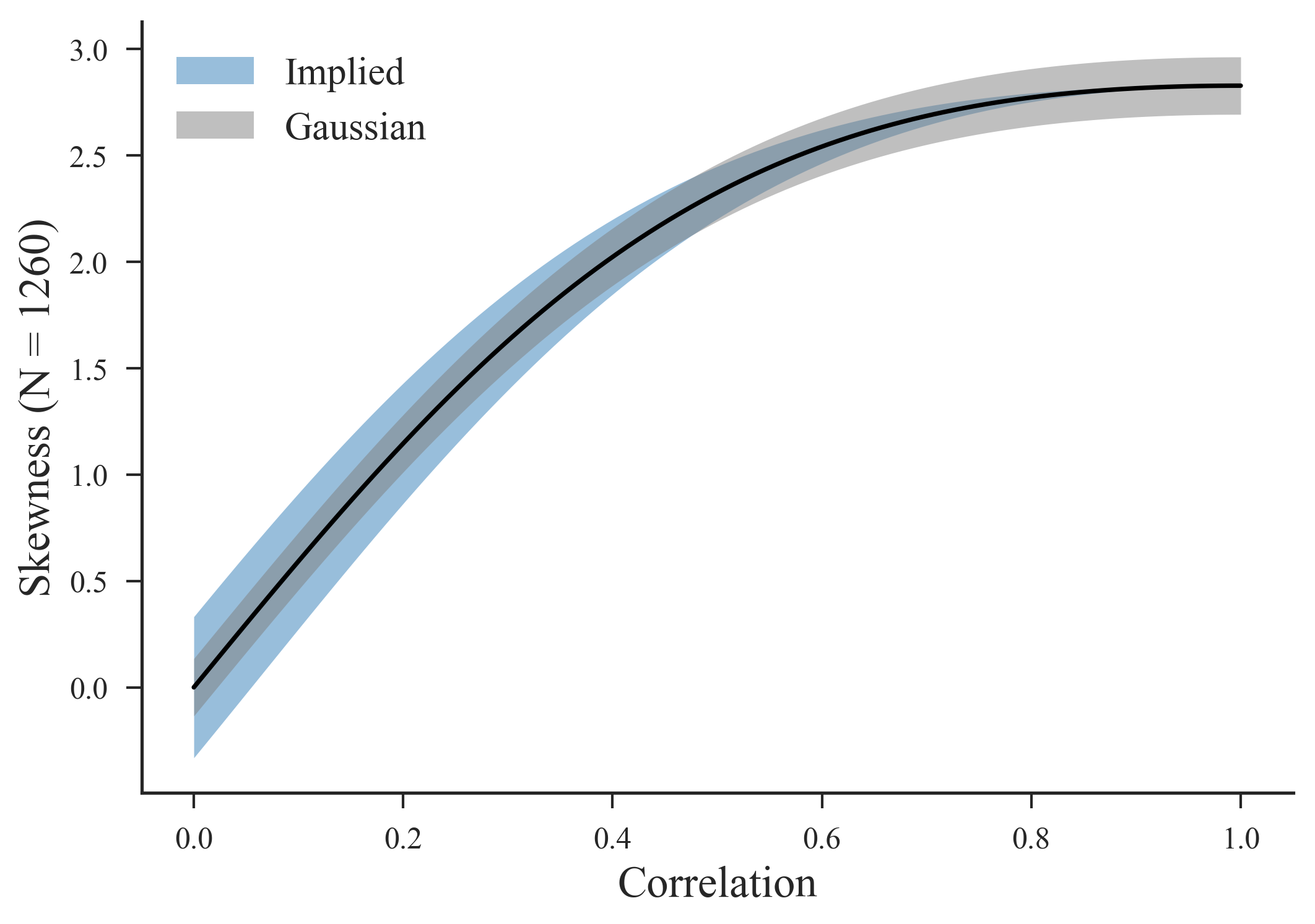}}
	\caption{{\bf Standard errors for skewness for different sample sizes, implied vs Gaussian} Implied standard errors, especially for skewness are generally larger than those for normal distributions. We argue that the   implied standard errors are more appropriate for dynamic strategies.
		\label{fig:StdErrSkewGauss} \label{fig:SkewCI1yr} }
\end{figure}

\begin{figure}
	\centering
	\subfloat{\includegraphics[width=0.33\linewidth]{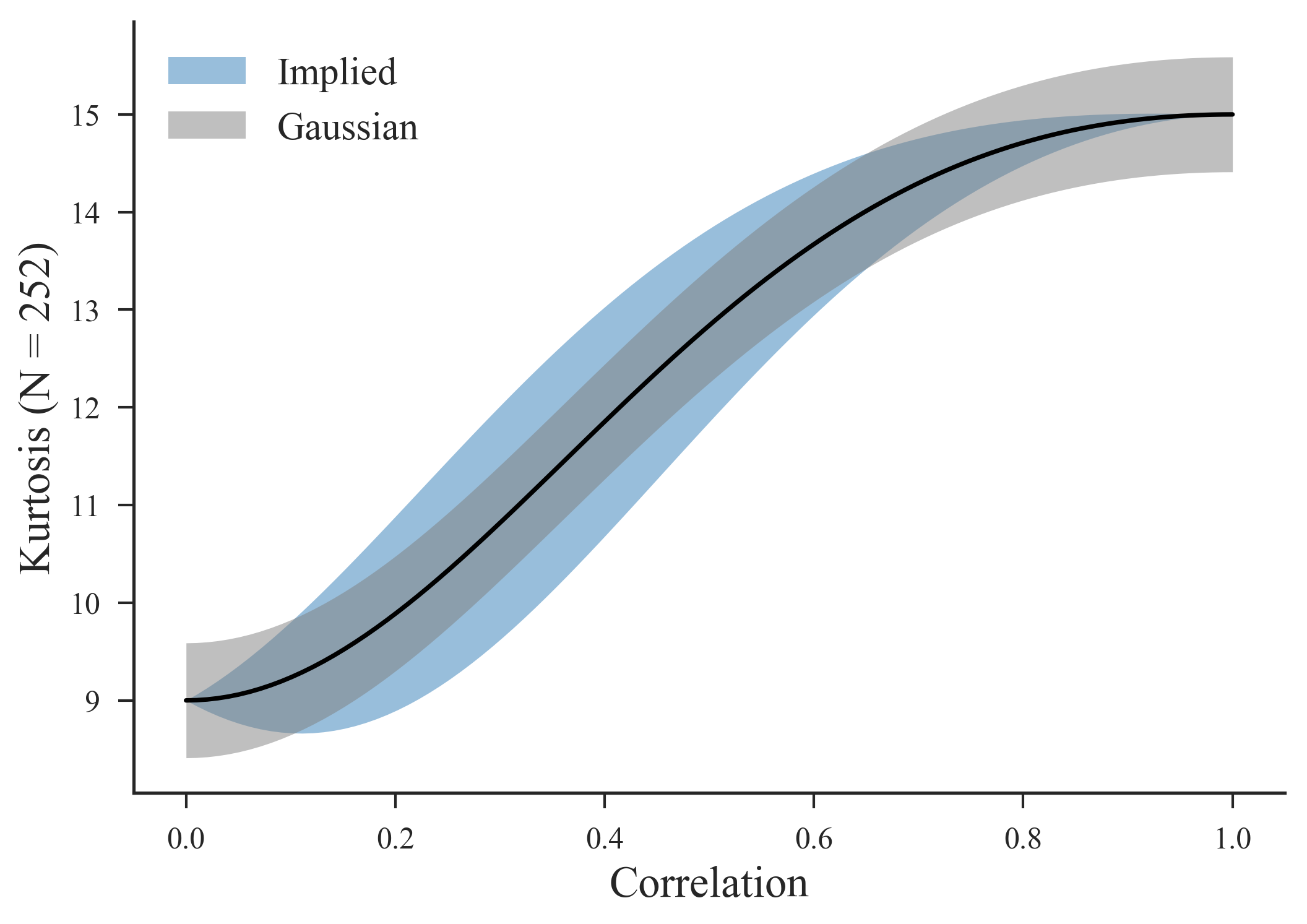}}
	\subfloat{\includegraphics[width=0.33\linewidth]{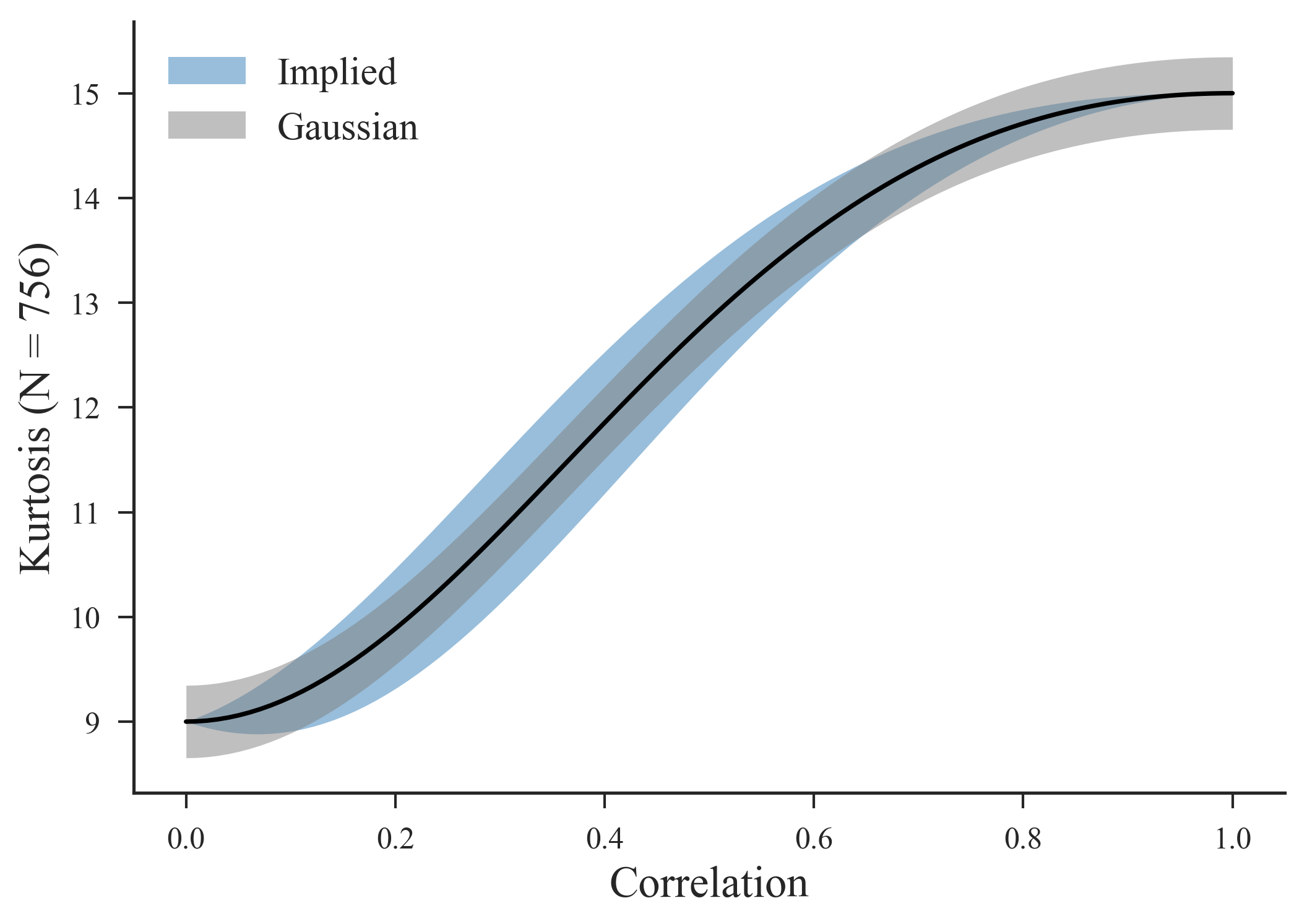}}
	\subfloat{\includegraphics[width=0.33\linewidth]{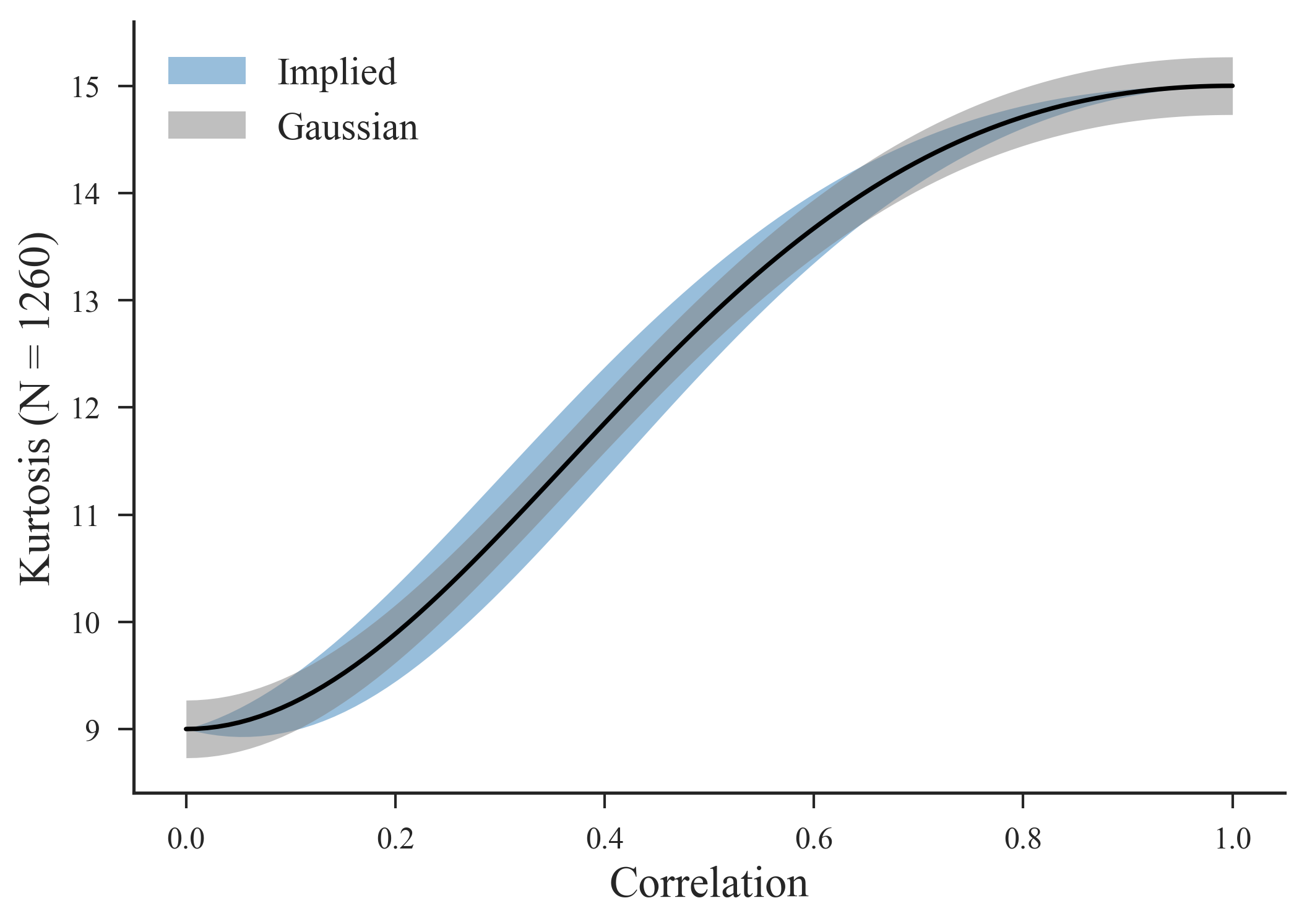}}
	\caption{{\bf Standard errors for kurtosis for different sample sizes, implied vs Gaussian} Implied  kurtosis standard errors are sometimes larger and sometimes tighter than the Gaussian case. We argue that the  {\em implied} standard errors are more appropriate for dynamic strategies.} \label{fig:StdErrKurtGauss} \label{fig:KurtCI1yr}
\end{figure}


While we can solve for $\rho$ in terms of $\gamma_k$ for $k=3,4$, the formulas are not easy to present (especially for kurtosis) and we believe that the statement, in terms of correlation is easier to use.

We note that, unlike the argument for using our refined standard errors over those presented in \cite{Lo}, the rationale for using the skewness and kurtosis standard errors presented in equations (\ref{Eqn:SkewKurtStderr}) is that returns are, for most practical purposes, not close to normal, and the product of two normals is more relevant for dynamic strategies. We elaborate on this in Section \ref{Section:Gauss vs Product}.


\section{Multiple assets}

We consider whether there is a diversification benefit from adding more independent {\sl bets} to our portfolio, and to what extent we can benefit from this. For context we note that portfolios of dynamic strategies can behave very differently from single strategies. For instance, Hoffman-Kaminski have noted (\cite{Hoffman}) that while single strategies can have skewness ranging from around $[1.3,1.7]$ and kurtosis from $[8.8,15.3]$, portfolio skewness can be as low as $0.1$.

We first consider $N$ indepedent returns as an N-vector, $R_t\sim {\mathscr N}(0,\sigma^2 I)$, assumed to have the same variance. We devise signals $X_t \sim {\mathscr N}(0,\gamma^2 I)$. The inner-product $X_t\cdot R_t$ has a density $\psi$ whose moment generating function is given by \cite{Simons}:

\[
M_N(t)=(1-2t\sigma\gamma\rho-\sigma^{2}\gamma^{2}t^{2}(1-\rho^{2}))^{-N/2}.
\]
From this we can easily derive four moments:
\begin{eqnarray*}
\mu_{1},  =&N \sigma\gamma\rho\\
\mu_{2}  =&N\sigma^{2}\gamma^{2}((N+1)\rho^{2}+1)\\
\mu_{3}  = &N (N+2)\sigma^{3}\gamma^{3}\rho((N+1)\rho^{2}+3)\\
\mu_{4}  =  &\sigma^{4}\gamma^{4} \Bigl( (N+6) (N+4)(N+2) N \rho^4+3 (N+2)N (1-\rho^2)^2 + \\
& \quad  6(N+f4)(N+2)N\rho^2(1-\rho^2) \Bigr)\\
\end{eqnarray*}
This leads to centralized moments 
$$\sigma^2 = N(\rho^2+1)$$
and 
$$\mu_3^c =  2 N \rho(\rho^2+3) $$

From these we derive the Sharpe ratio:

$$\SR =  \frac{\sqrt{N} \rho} { \sqrt{\rho^2 +1}}$$

Maximizing the SR over $\rho$ leads to $\frac{\sqrt{N}\sqrt{2}}{2}$, clearly showing the benefit of diversification when measuring the Sharpe ratio.

The skewness is
$$ \gamma_3 = \frac{1}{\sqrt{N}} \frac{2  \rho(\rho^2+3)} { ( \rho^2+1)^{3/2}} $$
and if we consider maximal Sharpe, the corresponding skewness is
$$ \gamma_3^{\max} = {8N  \over (2N)^{3/2}} = \frac{2\sqrt{2}}{\sqrt{N}}$$
will show reductions on the order of $1/\sqrt{N}$ in the total number of (orthogonal) assets. This is as expected from large diverse portfolios. In the limit, simple application of central limit theory should give us asymptotic normality.  Effectively, introducing more purely orthogonal assets will increase Sharpe ratios, but decreases the (relatively desirable) positive skewness. 

If we have multiple possibly correlated assets and multiple, possibly correlated signals, we assert that an optimal strategy would be to perform {\em canonical correlation analysis} (CCA),
\footnote[7]{Canonical correlation  (from \cite{Hotelling}, see for example, \cite{Rencher}) is defined by first finding the linear vectors $w_1$ and $v_1$ withe $|w_1|=|v_1|=1$, such that $\rho(w_1\cdot R, v_1\cdot X)$ is maximized. The resulting correlation is the {\em canonical correlation}. The {\em canonical variates} are defined by finding subsequent unit-vectors  $w_k$ and $v_k$ such that $\rho(w_k\cdot R, w_j\cdot R)=\delta_{kj}$, $\rho(v_k\cdot X, v_j\cdot X)=\delta_{kj}$, and $\rho(w_k\cdot R,v_k\cdot X )$ is maximized, leading to $\rho(w_k\cdot R,v_j\cdot X )=r_k\delta_{kj}$ . The solution is via a generalized eigenvalue problem
\begin{eqnarray*}
\Sigma_{RR}^{-1}\Sigma_{RX}\Sigma_{XX}^{-1}\Sigma_{XR}w_k&=& r_k^2 w_k\\
\Sigma_{XX}^{-1}\Sigma_{XR}\Sigma_{RR}^{-1}\Sigma_{RX}v_k &=& r_k^2 v_k
\end{eqnarray*}
where $\Sigma$ is the partitioned correlation matrix of $(R,X)$ and the canonical correlates $w_k$ and $v_k$ are the eigenvectors with the same eigenvalues $r_k$. The corresponding portfolios of {\em canonical strategies}, $S_k^{CCA}\equiv (v_k\cdot X)(w_k\cdot R)$ each have returns and variances as characterised by equation (\ref{eqn:centered_moments1} and \ref{eqn:centered_moments2}) with corresponding correlations $r_k$ (i.e., with Sharpe ratios given by $\SR[S_k]=r_k/\sqrt{r_k^2+1}$) and, due to their independence, can easily be weighted to optimize the portfolio Sharpe Ratio. 
The method of weighting the cannonical strategies is of course, similar  to a risk-parity portfolio, due to the independence of asset returns.
We assert that this method gives  the maximal Sharpe ratio for the linear combination of signals and returns, although we leave this proof to a subsequent paper.} resulting in a set of decorrelated strategies (using a and combination of signals to weight a portfolio of assets). The resulting strategies are  decorrelated but with unequal returns and variances. Many results of this section would apply after scaling  the portfolio returns. The end-result could easily be optimized using simple mean-variance analysis (reweighting the returns on the independent strategies). We leave the details for another study. 

While our optimizer is unlikely to be in use among CTAs, it is still notable that widely diversified CTAs (irrespective of underlying asset correlations) appear to have decent Sharpe ratios but relatively lower positive skewness, much in line with the discussion of this section. Our simple results here about the final Sharpe ratio and skewness of course depend on independence of the underlying assets and of course the signals themselves, which must only be correlated with their respective asset returns. While this is a not an altogether natural setting, it is suggestive of the gains that can be made in introducing purely orthogonal sources of risk, or perhaps in orthogonalizing (or attempting to) asset returns prior to forming signals,  later recombining into a portfolio, and that this may lead to far more desirable properties of portfolios than finding strategies on multiple non-orthogonalized assets.

\section{Gaussian Returns vs Products of Gaussians Returns}
\label{Section:Gauss vs Product}

While we believe that the assumption of Gaussian returns (and Gaussian signal) is a simplification, we also believe this is far more realistic than the assumption of Gaussian returns for a dynamic strategy.  Throughout this paper we consider Gaussian (log) returns $R\sim\mathscr{N}(0,\sigma_R^2)$ and Gaussian signal $X\sim\mathscr{N}(0,\sigma_X^2)$ which together are jointly Gaussian, and together form components of the dynamic strategy $S_t= X_t R_t$, whose properties we study.

To be clear, our signal is not considered to have foresight and is fully known as of time $t$, while the return $R_t$ is from $t$ to $t+\delta t$. All expectations calculated are unconditional, or, can be thought of as conditioned on $t_0<t<t+\delta_t$. Consequently, each element, the signal and the return will be random variables.

Were we to consider expectations conditional on $t$, then the resulting strategy returns  $S_t$ would be trivially Gaussian. In the unconditional case, the resulting returns are far more interesting and relevant.

CTA returns are known to generally be positively skewed and highly kurtotic over the relevant horizons we are concerned with (i.e., daily, weekly, monthly), as has been noted by \cite{Potters-Bouchaud}, \cite{Hoffman} and others. If we measure far longer-horizon returns, asymptotic theory should show that  favourable qualities like skewness may disappear.

Consequently, even though we make many comparisons to results stemming from either asymptotic theory (e.g., \cite{Lo}) or using exact normality, this comparison does not, in fact, compare like-for-like. Clearly \cite{Lo} is appropriate for large-samples, as is possible under conditions when the central limit theorem (CLT) holds, e.g., with weak-dependence, summing returns over increasingly longer  horizons, or in the case of a large cross-sectional dimension with increasing numbers of decorrelated assets. For dynamic strategies, asymptotic normality should be expected for large numbers of decorrelated dynamic strategies as well as for long-horizon (e.g., annual or longer, non-overlapping) returns for single dynamic strategies.

Consequently, we believe our standard error results are more appropriate for hypothesis testing on statistics for dynamic strategies. We discuss a strategy for establishing product measures as large-sample limits in appendix \ref{Section:FullDist}, although asymptotics are beyond the scope of this current study.

\section{Conclusion}

Fully systematic dynamic strategies are used by a large portion of the asset management industry as well as by many non-institutional participants.  Meanwhile, they are only partly understood. 
Many funds and strategies (e.g., especially investment bank {\em smart-beta} or {\em styles-based} products) involve investment in strategies which are not optimised in any sense. Strategies which are paid via index-swaps have great limits in terms of their adaptability, leading to often highly suboptimal end-results.  While there have been some very significant results derived in the theoretical properties of these dynamic strategies, there is still much more work left to do. Given that most  academic literature in this area considers more general distributions, there has not been a firm foundation to build and extend these results. 

It is hoped that this paper does form a foundational approach to the study of dynamic strategies and {\em how} to optimize them. We make efforts to understand their properties without claiming to understand {\em why} they work (i.e., why there are stable ACFs in the first place). Given that most asset returns returns are known to have non-trivial autocorrelations, we can establish many results.
In particular, we have derived a number of results merely by applying well-known techniques to dynamic strategies, e.g.,:
\begin{itemize}
	\item Strategy returns can be shown to be positively skewed and leptokurtic.
	\item Sharpe ratios can be characterized, as can skewness and kurtosis.
  \item The standard errors for Sharpe, skewness and kurtosis can be derived.
  \item Strategies designed to optimise Sharpe ratios should be based on TLS rather than minimizing prediction error.
	\item Gains from adding orthogonal assets/risks can be quantified.
  
 \end{itemize}
Some of these items are empirically well-known, but others are genuinely new. Meanwhile, we have extended our results to the derivation of over-fitting penalties akin to Mallow's $C_p$ or AIC and can be used to do model selection and predict likely out-of-sample Sharpe ratios from in-sample fits (see \cite{Firoozye2}).

Our study is incomplete. We believe that there is a good deal of interesting work to be done in areas such as:
\begin{itemize}
  \item optimal linear strategies incorporating transaction costs.
	\item optimal linear strategies relaxing normality.
	\item normalized linear signals (e.g., z-scores) and optimal non-linear functions of z-scores.\footnote[8]{We note that normalized signals applied to normalized returns series can be represented as the product of two Student t-distributions, which is also relatively well-studied \cite{Multi-T, Joarder} and the results are qualitatively very similar to those which we have produced in this study. However, the more commonly used strategy of applying normalized signals to returns, with the resulting strategies then vol-scaled, cannot be derived as a trivial application of well-known results}
  \item non-linear strategies which are optimised to specific utility functions, possibly incorporating  smoothness constraints, especially when relaxing normality.
	\item local optimality when relaxing stationarity.
	\item good-deal bounds in the presence of auto-correlated assets with possible non-stationarity or structural breaks.
\end{itemize}

We note that our assumptions were never meant to be completely realistic: stationary returns with fixed ACF and Gaussian innovations can only work in theory, not in reality. Many quantitative traders design strategies to overcome the challenges of dealing with real-world data issues and the issues of over-fitting.
We nonetheless present them as a good starting point for further analysis, hoping to use this work as the basis for further exploration and to put the general study of dynamic strategies onto a more firm theoretical footing.

Some of our findings should be of note to practitioners.  In particular, the use of OLS and other forecast error minimizing methods is not necessarily optimal, depending on the problem at hand; total-least squares or other correlation-maximizing methods such as CCA may be more efficient. High Sharpe ratios and positive skewness are often quoted as rationales for entering into strategies and, strategies are changed with the rationale of increasing these measures. The relative significance of any of these changes depends on confidence intervals or standard errors, and we have derived these specifically suited for dynamic trading strategies. Kurtosis is not studied as often, but as we show, all dynamic strategies should be leptokurtic and this is an important attribute of these strategies. Other results, such as over-fitting penalties and optimal non-linear strategies, we save for later papers. With a more solid theoretical footing as a sort of {\em rule-of-thumb} for the development, optimisation, selection and alteration of dynamic strategies, we only hope that there can be room to improve strategy design.


\vspace{30pt}
\section*{Acknowledgements}
N.\ Firoozye would like to give his wholehearted love and appreciation to Fauziah, for hanging on, when the paper was always {\em almost done.} I am hoping the wait is finally over. 
Adriano Soares Koshiyama would like to to acknowledge the funding for its PhD studies provided by the Brazilian Research Council (CNPq) through the Science Without Borders program.

The authors would also like to thank Brian Healy and Marco Avellaneda for the many suggestions and encouragement. Finally, were it not for the  product design method as practised by Nomura's QIS team, the authors would never have been inspired to pursue a mathematical approach to this topic.

\textemdash \textemdash \textemdash \textemdash \textemdash \textendash{}
\tabularnewline

\appendix

\section{Full distributions for single period}
\label{Section:FullDist}

In general, for $X$ and $R$ having joint density $\psi^{X,R}(x,r)$, and
have $S_t= X_{t}R_{t}$ is known to have the product pdf, 

\begin{equation}
\label{eqn:product}
\psi^S(s)=\int_{-\infty}^\infty\psi^{X,R}\bigl(x,\frac{s}{x}\bigr)\frac{1}{|x|} dx
\end{equation}

and, in the special case where $X\sim\mathscr{N}(0,\sigma_X^2)$ and $R\sim\mathscr{N}(0,\sigma_R^2)$ jointly normal with correlation $\rho$ (i.e., $\psi$ being a bivariate gaussian), this results in the closed-form expression:

\begin{equation}
p_{s}=\frac{1}{\pi\sigma_R\sigma_X}\exp\bigl(\frac{\rho s}{\sigma_R\sigma_X(1-\rho^{2})}\bigr)K_{0}\bigl(\frac{|s|}{\sigma_R\sigma_X(1-\rho^{2})}\bigr)
\label{eqn:FullDist}
\end{equation}
where $K_0(\cdot)$ is a modified Bessel function of the $2^{nd}$ kind  (\cite{Simons}, p 51, eq 6.15). The more general density for non-zero means, is given in \cite{Exact} as an infinite series. In the special cases of independence and of correlated but zero mean, the expressions become much simpler and we choose to focus on the zero-mean case here. The density is unbounded at zero and has fat tails and positive skewness, becoming more pronounced with higher correlation. We can see the distribution for a variety of correlations in figure (\ref{fig:exact}), with the skewness becoming increasingly pronounced for higher $\rho$. In the limit as $\rho\rightarrow 1$ the distribution converges to that of the central $\chi^2$ distribution with one degree of freedom.

\begin{figure}[h!]
	\begin{center}
		\begin{minipage}{110mm}
			\includegraphics[width=\linewidth]{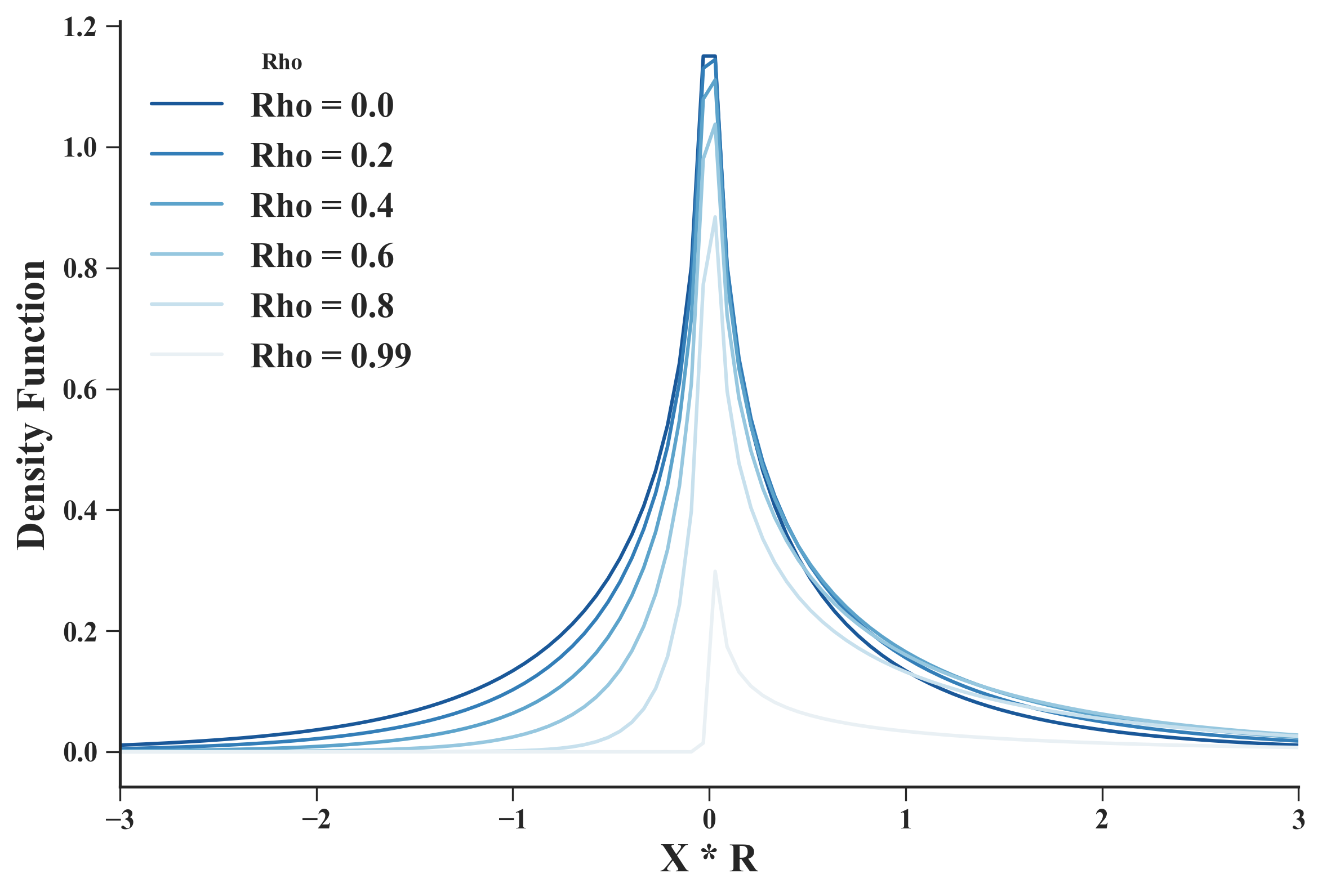}
			\caption{{\bf Complete product distributions}, for $\rho\in\{0,0.2,0.4,0.6.,0.8,1.0\}$, normalised to have unit variance, so they can be depicted on one plot. Note the singularity at 0, the increasing asymmetry, nearly truncated left-tails and marginally fatter right-tails with increasing $\rho$.} \label{fig:exact}
		\end{minipage}
	\end{center}
\end{figure}

In fact, $K_0(z) = O(e^{-z}/\sqrt{z})$ for $z\rightarrow \infty$ and we can see that the tail behaviour of the pdf in equation (\ref{eqn:FullDist})  changes quite significantly from when $\rho=0$  and $K_0(z)$ is the only term to consider, to when $\rho>0$, introducing an asymmetry. The Bessel function is unbounded at $z=0$. Asymptotically, we have the following behaviour:

\begin{eqnarray*} 
p_s =& O(e^{-|s|}/\sqrt{|s|}) \quad&\mbox{for~} \rho=0, |s|\rightarrow \infty\\
p_s =& O( e^{-s}/\sqrt{|s|}) \quad&\mbox{for~} \rho>0, s\rightarrow \infty\\
p_s =& O( e^{s}/\sqrt{|s|}) \quad&\mbox{for~} \rho>0, s\rightarrow -\infty\\
p_s =& O(-\log{|s|}) \quad&\mbox{for~}  |s|\rightarrow 0 \\
\end{eqnarray*}

\section{Convolution Filters as Jointly Gaussian}

\label{Section:ConvolutionFilters}

If we have a purely-random mean-zero covariance-stationary discrete-time Gaussian process $R_t$, we note by Wold Decomposition, that all stationary Gaussian processes can be represented as MA$(\infty)$ in terms of Gaussian innovation process and coefficients in $\mathit{l}^2$, with no deterministic component, i.e., 
$$ R_t = \sum_{k=0}^\infty \phi(k)\epsilon_{t-k}$$
for $\epsilon\sim\mathcal{N}(0,\sigma^2)$, $\sum_0^\infty\phi_k^2<\infty$ and $\phi(0)=1$.

More specifically, we have
\begin{eqnarray*}
E[R_t] &=&0\\
Var(R_t)&=&\sigma_R^2\\
corr(R_t,R_s) &=& \gamma(t-s)
\end{eqnarray*}
(i.e., with ACF $\gamma$), and this would be sufficient to determine $\phi$ if we so wished.

We are interested in constructing signals: $X_t$. A standard signal we will consider is  a convolution signal, i.e.,

$$X_t = \sum_{k\geq 1} \phi(k) R_{t-k}$$

All the signals mentioned in the introduction (e.g., moving average or difference of moving averages or  ARMA based forecasts), can be expressed as convolutions with historic returns. A convolution filter is an example of a time-invariant linear filter. It the coefficients $\phi \in \mathit{l}^2$ then it is well known that the resulting filtered series $X_t$ are Gaussian\footnote[9]{see e.g., Gallagher, R, {\em Stochastic Processes: Theory for Applications}, 2014, (Cambridge UP: Cambridge), or Gallagher R,   Principles of Digital Communications. {\em MIT Open Coursework}. Section 7.4.2, Theorem 7.4.1.}. The filtered series $X_t$ is also jointly Gaussian with $R_t$.

\begin{eqnarray*}
E[X_t] &=& 0\\
Var(X_t) &=& \sum_{k,j\geq 1} \phi(k) \phi(j)E[R_{t-k} R_{t-j}]\\
                &=&\sum_{k,j\geq 1} \phi(k)\phi(j)\gamma(k-j)\sigma_R^2
\end{eqnarray*}
(dropping all first order terms because $E[R_t]=0$ )
and,
\begin{eqnarray*}
corr(R_t,X_t) &=& \frac{E[\sum_{k\geq 1} \phi(k) R_{t-k}, R_t]}  {std(X) \sigma_R}\\
                      &=& \frac{\sum_{k\geq 1} \phi(k)\gamma(k)}{ ( \sum_{k,j\geq 1} \phi(k)\phi(j)\gamma(k-j) )^{1/2}}
\end{eqnarray*}
cancelling out all $\sigma_R$ terms.

Consequently,
$$sgn(corr(R_t,X_t) ) = sgn( \gamma\cdot \phi))$$
(i.e., the sign of this infinite inner product matters most for determining usefulness of a given convolution design).

Of the signals mentioned in the introduction, EWMA and SMA in returns, differences of EWMAs and SMAs in returns, and forecasts from ARMA models are all examples of convolution filters with $\mathit{l}^2$ coefficients. Most signals constructed in levels (e.g., the difference between a price and its simple moving average), are not, in general, Gaussian, although a difference between a price and one or more EWMAs may be Gaussian depending on the data-generating process for the price series (i.e., for MA processes). 

Of course, a linear time-invariant filter with $\mathit{l}^2$ coefficient is just one example of a signal $X_t$ which is jointly Gaussian with returns $R_t$. Similarly, if $Z_t$ is a set of Gaussian (exogenous) features, then $X_t  = Z_t \beta$  will also be Gaussian and we will assume the $Z_t$ are jointly Gaussian with $R_t$, meaning also $X_t$ and $R_t$ will be jointly Gaussian.

\section{Limiting behaviour for convolution of stationary returns}

We assert some asymptotic approximation results for dynamic strategies, only outlining  their proof. Our claim is that this  justifies the use and analysis of product of Gaussian distributions in stationary (or locally stationary) distributions. The proof itself is the direct consequence of much more general work on the limits of quadratic forms by G\"otze and Tikhonov and by the Wold decomposition theorem.

Letting  $\eta$ be iid random variables with mean zero and unit variance, and letting $\epsilon$ be iid normal random variables with zero mean and unit variance, we form the quadratic forms:
$$Q_n=\sum_{j,k=1}^n a_{jk}^n \eta_j\eta_k \;\hbox{and}\; G_n=\sum_{j,k=1}^n a_{jk}^n \epsilon_j\epsilon_k.$$ We write the  metric
$$\delta_n(Q_n,G_n)=\sup_x|P\{Q_n\leq x\} - P\{G_n\leq x\}|.$$ We simplify the statement of Theorem 1 from \cite{Goetze1}:

\begin{theorem}[Goetze-Tikhomirov] Let $\eta$ be IID with
 $$E\eta=0,\; E \eta^2=1, \;E|\eta|^3=\beta_3<\infty.$$ Then there is a constant $C$ such that 
 $$\delta_n(Q_n,G_n)\leq C\beta_3^2 \Gamma_n$$
 where $\Gamma_n=\max_{1\leq j\leq n} \sum_{k=1}^n |a_{jk}^n|$.
 \end{theorem}
Our assertion is a simple application of the results in    \cite{Goetze1}, (see \cite{Goetze2} and \cite{Goetze3} for further results) which applies to limiting theorems of quadratic forms of random variables.  

\begin{theorem}[Products of Gaussians] Let $R_t$ be a covariance stationary process with bounded 3rd moments and mean zero and its Wold decomposition given by $R_t=\sum_{s=1}^\infty b(s)\eta(t-s)$ with $\eta$ a white-noise process. Let the signal $X_t$ be a convolution of the lagged returns $R_t$ with an $\mathcal{L}^2$ convolution kernel, $\phi$ and $X_t=\sum_1^\infty \phi(s)R_{t-s}$.
We let $R_t^N=  \sum_0^N b(s) \eta(t-s)$ and $X_t^N = \sum_1^N \phi(s)R_{t-s}^N$ be truncated  sums (only involving the first $N$ terms),
$$S_t^{N}=\frac{1}{\sqrt{N}}X_t^N\cdot \frac{1}{\sqrt{N}} R_t^N$$
be the scaled truncated strategy returns.

Then there is a pair of Gaussians $\tilde{R}^{N}_t$ and $\tilde{X}^{N}_t$ ( $\tilde{S}_t^N=\tilde{X}^N_t\cdot \tilde{R}^N_t$ be the Gaussian strategy returns)such that $$\delta_n(S^N_t,\tilde{S}^N_t)\rightarrow 0,$$
or, in other words, that the product of Gaussian approximation can be arbitrarily close to the original strategy.

\end{theorem}



We note that the product $S_t=X_tR_t$ is given by the quadratic form: 
$$S_t=X_t R_t = <A\eta,\eta>$$

where $A$ is the operator given by $$A(u,v)=\sum_{s\leq t-1} \phi(t-s) b(s-u) b(t-v)
$$
for $u,v\leq t$. 
\begin{eqnarray*}
S_t^{N}&=&\frac{1}{\sqrt{N}}X_t^N\cdot \frac{1}{\sqrt{N}} R_t^N\\
&=&<A^{N} \eta^N,\eta^N>
\end{eqnarray*}
where $A^n$ is an $n\times n$ matrix
 $$A^n_{u,v} =\frac{1}{N}\sum_{s\in[t-n,t-1]} \phi(t-s) b(s-u) b(t-v),$$ for $u,v$ ranging in $[t-n,t]$ and $\eta^N=\{\eta_s\}_{s\in[t,t-N]}$.


 We note that $A^n$ is lower triangular with  no diagonal terms (elements on the diagonal correspond to instantaneously available knowledge, contemporaneous with the observed returns themselves and elements in the upper triangle of the matrix correspond to direct foresight). Moreover, with sufficient conditions on the original series $R_t$ (i.e., on the Wold coefficients $b$) and on the convolution coefficients $\phi$, the $\Gamma_N=\max_{1\leq j\leq n} \sum_{k=1}^n |A_{jk}^n|$ can be shown to decay to zero.

A direct application of the theory of quadratic forms would apply when the convolution coefficients are sufficiently well-behaved at infinity.

This is only one of the possible approaches to an asymptotic theory justifying the use of products of Gaussians.\footnote[10]{Other approaches include assuming infinitessimal Gaussian increments which are observed and ``stored'' and used in a convolution, then applied as a weight on a strategy which itself is held for a longer time. This effecftively results in some product of averages of returns and, obviously, when appropriately scaled can be shown to have a limit of a product of Gaussians.} While asymptotic approaches are not the main point of this paper, it should be clear that products of Gaussians help to approximate the behaviour of a wide array of dynamic strategies.

\section{ Nonzero means: Sharpe ratios and Skewness}

By an abuse of notation, we define $\SR[R]$ to be $\mu_R/\sigma_R$ and by an abuse of notation, we define $\SR[X]=\mu_X/\sigma_X$ (for $X$ the signal),

{\bf Corollary 1:} If $R\sim \mathscr{N}(\mu_R,\sigma_R^{2})$ and $X\sim \mathscr{N}(\mu_X,\sigma_X^{2})$ then
$$ \SR[X\cdot R] = {\SR[R]\cdot \SR[X] + \rho \over (\SR[R]^2+\SR[X]^2+2\rho \SR[R]\cdot \SR[S]+\rho^2+1)^{1/2}}
$$

{\bf Corollary 2:} If $R\sim \mathscr{N}(\mu_R,\sigma_R^{2})$ and $X\sim \mathscr{N}(\mu_X,\sigma_X^{2})$  then
$$ \gamma_3[X\cdot R] = { 2\rho (\rho^2 + 3 + 3 \SR[R]^2+3 \SR[X]^2)\over
(\SR[R]^2+\SR[X]^2+2\rho \SR[R]\cdot \SR[X]+\rho^2+1)^{3/2}}
$$

We note  the one period Sharpe ratio of the   strategy may depend on both the interaction between the Sharpe ratios of the Signals (weights) and the Returns, in particular whether they have the same sign or not, together with the sign of the correlation. In fact, the amplitude of the resulting strategy SR may be more dependent on the respective Sharpe ratios rather than $\rho$ since after all, $-1\leq\rho\leq 1,$while $\SR[R]$ and $\SR[X]$ may individually be above $1$.

\section{Transaction Costs}
The sections above consider optimal linear strategies with no transaction costs. If we include transaction costs then the formulas are not nearly as elegant, but the results may still remain tractable.

Maximizing Sharpe ratios are often the result of maximizing a quadratic utility of returns, e.g., 
\begin{equation}
U[XR]=E[XR] - \gamma \var[XR] \label{quadutil}
\end{equation}
where $\gamma$ is a measure of risk-aversion, sometimes called a {\em Kelly constant}. Extremals of the utility in equation (\ref{quadutil}) are known to coincide with maximal Sharpe ratios.

We only look at  convolution filter strategies, i.e., $\phi=(0,\phi_1,\phi_2,\ldots,\phi_K)$ which give a corresponding signal as $X_t=\phi * R_t =\sum_1^K \phi_k R_{t-k}$. As we mentioned above, fitting $\phi$ via TLS instead of OLS is most appropriate in the case of no-transaction costs.

If we include transaction costs proportional to a constant $\nu$, rather than to maximize a quadratic utility in (\ref{quadutil}), we can add  the extra term\footnote[11]{Alternatively,  a term such as $E[|\Delta X|\cdot P]$ where $P=P_0+\sum R_t$ could be added. Again, with work we could equally well characterize this expectation, using properties of distributions derived from Gaussians and some application of Isserlis' theorem}, e.g.,
$$U[XR]=E[XR] - \gamma \var[XR] -\nu E[|\Delta X|]$$
Given that $$\Delta X = \phi_1 R_{t-1} + \sum_{k=1}^K \Delta \phi_k R_{t-k} - \phi_K R_{t-K}\equiv\Delta\phi *R $$  where $\Delta \phi = (0,\phi_1,\Delta\phi_1,\Delta\phi_2,\ldots,\Delta\phi_K,-\phi_{K})$.
The r.v. is normal, $\Delta X\sim \mathcal{N}(0,\sigma_{\Delta X}^2)$ and, using the properties of {\em folded Gaussian variables}, we can characterise
$$E[|\Delta X|]=\sqrt{\frac{2}{\pi}}\sigma_{\Delta X}$$
The entire utility then can be written as
$$U[XR]= \rho\sigma_X\sigma_R - \gamma \sigma_X^2\sigma_R^2 (1+\rho^2) - \nu\sqrt{\frac{2}{\pi}}\sigma_{\Delta X}$$

Optimising this utility will be very much like a standard least-squares problem except the term $\sigma_{\Delta X}$ is a form of regularization.

In fact if we let $C$ being the ACF (Toeplitz) matrix of $(R_t,\ldots R_{t-k})$, i.e.,
$$C=
\begin{bmatrix}
    1 & c(1) & c(2) & \dots  & c(k-1) \\
    c(1) & 1 & c(1) & \dots  & c(k-2) \\
    \vdots & \vdots & \vdots & \ddots & \vdots \\
    c(k-1) & c(k-2) & c(k-3) & \dots  & 1
\end{bmatrix} $$ 
and $X=\phi*R$ with $\phi=(\phi_1,\phi_2,\ldots,\phi_k)$ and let $\mathbbm{1}_0=(1,0,\ldots,0)$
then $\sigma_X=\sigma_R\sqrt{\phi'\cdot C\cdot\phi}$, $\rho=\phi'\cdot C\cdot \mathbbm{1}_0 $ and 
$\sigma_{\Delta X}
=\sigma_R\sqrt{(\Delta\phi)' \cdot C\cdot (\Delta\phi)}$, effectively penalizing changes in the $\phi_k$.

The resulting optimisation problem thus becomes
$$U[XR] = \sigma_R^2 \sqrt{\phi'\cdot C\cdot \phi} \phi'\cdot C\cdot\mathbbm{1}_0-\gamma\sigma_R^4 (\phi'\cdot C\cdot \phi)(1+(\phi'\cdot\mathbbm{1}_0)^2)-\nu\sqrt{2\over\pi}\sigma_R\sqrt{(\delta\phi)'\cdot C\cdot (\Delta\phi)}$$

This final regularization term should ensure that the filter weights $\phi_k$ do not vary too much between themselves (i.e., it is a sort of  smoothness constraint analogous to those in  a {\em Lasso} or {\em Ridge-regression}, but with a slightly different functional form).  Unlike the case of an $\mathscr{L}^2$ penalty as in Ridge regression or an $\mathscr{L}^1$ penalty as in Lasso, this term though is neither linear nor quadratic.

We do not consider properties of the solutions of optimal trading strategies with transaction costs in this paper.

\section{Multiperiod Returns}

Given the ease of analysis of Gaussian returns, it is straightforward to calculate moments of the strategy returns to any horizon. While we do not explore further implications, we produce relevant formulas in this section for future elaboration.

For long-horizon trades we note the following (\cite{Magnus})

{\bf Theorem}[Magnus] Let $A$ be a symmetric matrix and $R\sim \mathscr{N}(0,V)$ with $V$ positive definite. Define $p=R'A R$. then the expectation, variance, skewness and kurtosis of $p$ are:
\begin{eqnarray*}
\mu &= &\tr(AV)\\
\sigma^2 &=& 2 \tr(AV)^2\\
\gamma_3 &=& 2\sqrt{2} \frac{\tr(AV)^3}{(\tr (AV)^2)^{3/2}}\\
\gamma_4 &=& 12 \frac{\tr(AV)^4}{(\tr(AV)^2)^{2}}
\end{eqnarray*}

which would allow us to calculate Sharpe ratios, skewness and kurtosis to any horizon. Continuous analogues are feasible using functional central limit theory for Wick products (see  \cite{FCLTWick}).

Given this and the various moment conditions for our Gaussian returns:
\begin{eqnarray*}
E[R_t]=&0\\
E[R_t^2]=&\sigma^2\\
E[R_t^3]=&0\\
E[R_t R_s]=&C(t-s)\sigma^2
\end{eqnarray*}

where $C(0)=1$ and $C(-k)=C(k)$, we can combine for characterising strategy moments.

If the ACF matrix $\tilde{C}$ is known with certainty, of course, then the linear filter which maximizes the correlation of signal to returns is merely given by finding the eigenvector corresponding to the smallest eigenvalue, i.e., 
$$v_{min}=argmin_v {v'\tilde{C} v\over |v|^2}$$ and normalizing the first coefficient to be one, i.e., $a_k=-v(k+1)/v(1)$.

For longer horizons, w use the formulas given by Magnus, or equally compute the term-structure by hand:
$$\mu_1(T) = E\sum X_t R_t = \sigma_X\sigma_R \rho T$$
and 
\begin{eqnarray*}
E(\sum X_t R_t)^2 &=& \sum\sum E[X_t R_t X_s R_s]\\
&=&\sum\sum\bigl(E[X_tR_t]E[X_sR_s]+E[X_tX_s]E[R_tR_s]+E[X_tR_s]E[X_sR_t])\\
&=& \sigma_X^2\sigma_R^2\bigl(T^2\rho^2+\sum\sum (C(i-j)D(i-j)+\rho(i-j)\rho(j-i))\bigr)
\end{eqnarray*}
where $C(k)$ is the ACF for $R$ and $D(k)$ is the ACF for signal $X$, and $\rho(k)=E[X_tR_{t-k}]$ and 
$\rho(k)=\rho(-k)$ and $\rho(0)=\rho$ is the contemporaneous correlation.

Consequently, 
$$ var(\sum X_t R_t) = \sigma_X^2\sigma_R^2\bigl(2\sum_{k=1}^{T-1} (T-k) (C(k)D(k)+\rho(k)\rho(-k)) + T(1+\rho^2)\bigr)$$

and consequently, the Sharpe ratio to any horizon is given by
$$SR(T) = \frac{\rho T^{1/2}}{1+\rho^2+2\sum_1^{T-1}\frac{T-k}{T}(C(k)D(k)+\rho(k)\rho(-k))}$$
giving us the term-structure of Sharpe ratios by horizon.

\section{Set-up details} \label{set-up}

If we have a purely-random mean-zero covariance-stationary discrete-time Gaussian process $R_t$, we note by Wold Decomposition, that all stationary Gaussian processes can be represented as MA$(\infty)$ in terms of Gaussian innovation process and coefficients in $\mathit{l}^2$, with no deterministic component. 

Specifically, we have
\begin{eqnarray*}
	E[R_t] &=&0\\
	Var(R_t)&=&\sigma_R^2\\
	corr(R_t,R_s) &=& \gamma(t-s)
\end{eqnarray*}
(i.e., with ACF $\gamma$)

Then we are interested in constructing signals: $X_t$. A standard signal we will consider is  a convolution signal, i.e.,

$$X_t = \sum_{k\geq 1} \phi(k) R_{t-k}$$
This is an example of a time-invariant linear filter. It the coefficients $\phi \in \mathit{l}^2$ then it is well known that the resulting filtered series $X_t$ are Gaussian\footnote[1]{see e.g., Gallagher, R, Stochastic Processes: Theory for Applications, Cambridge UP, 2014, or Gallagher R,  MIT Open Coursework, Principles of Digital Communications, Section 7.4.2, Theorem 7.4.1.}. The filtered series $X_t$ is also jointly Gaussian with $R_t$.

We note that if the $\phi(k)$ can be derived as the coefficients of an ARMA model forecast, or they can be from a simple EWMA, as we have mentioned in the paper.

\begin{eqnarray*}
	E[X_t] &=& 0\\
	Var(X_t) &=& \sum_{k,j\geq 1} \phi(k) \phi(j)E[R_{t-k} R_{t-j}]\\
	&=&\sum_{k,j\geq 1} \phi(k)\phi(j)\gamma(k-j)\sigma_R^2
\end{eqnarray*}
(dropping all first order terms because $E[R_t]=0$ )
and,
\begin{eqnarray*}
	corr(R_t,X_t) &=& \frac{E[\sum \phi(k) R_{t-k}, R_t]}  {std(X) \sigma_R}\\
	&=& \frac{\sum \phi(k)\gamma(k)}{ ( \sum_{k,j\geq 1} \phi(k)\phi(j)\gamma(k-j) )^{1/2}}
\end{eqnarray*}
cancelling out all $\sigma_R$ terms.

Consequently,
$$sgn(corr(R_t,X_t) ) = sgn( \gamma\cdot \phi))$$
(i.e., the sign of this infinite inner product matters most for determining usefulness of a given convolution design). A linear time-invariant filter with $\mathit{l}^2$ coefficient is just one example of a signal $X_t$ which is jointly Gaussian with returns $R_t$. Similarly, if $Z_t$ is a set of Gaussian (exogenous) features, then $X_t  = Z_t \beta$  will also be Gaussian and we will assume the $Z_t$ are jointly Gaussian with $R_t$, meaning also $X_t$ and $R_t$ will be jointly Gaussian.

\bibliographystyle{jasaauthyear}
\bibliography{xampl}


\end{document}